\documentclass[english]{elsarticle}

\makeatletter
\def\ps@pprintTitle{%
 \let\@oddhead\@empty
 \let\@evenhead\@empty
 \def\@oddfoot{}%
 \let\@evenfoot\@oddfoot}
\makeatother

\usepackage{amsmath,amsthm}
\usepackage{mathtools}
\usepackage{algorithm}
\usepackage{algpseudocode}
\usepackage{euscript}
\usepackage{enumerate}
\usepackage{enumitem}
\usepackage{subfigure}
\usepackage{tikz}
\usepackage{verbatim} 
\newcommand*{\LONG}{}

\newcommand{\LP}{\textit{\tiny LP}}

\algnewcommand{\IIf}[1]{\State\algorithmicif\ #1\ \algorithmicthen}
\algnewcommand{\EndIIf}{\unskip\ \algorithmicend\ \algorithmicif}

\algnewcommand{\IiIf}[1]{\State\algorithmicif\ #1\ \algorithmicthen}
\algnewcommand{\EndIiIf}{\unskip\ \algorithmicend\ \algorithmicif}

\algnewcommand{\IfThenElse}[3]{% \IfThenElse{<if>}{<then>}{<else>}
 \State \algorithmicif\ #1\ \algorithmicthen\ #2\ \algorithmicelse\ #3}
\algnewcommand{\EndIfThenElse}{\unskip\ \algorithmicend\ \algorithmicif}

\usetikzlibrary{arrows,decorations.pathmorphing,backgrounds,positioning,fit,petri}
\newtheorem{theorem}{Theorem}
\newtheorem{prop}[theorem]{Proposition}

\def\eps{\varepsilon}
\def\real{\hbox{\rm\vrule\kern-1pt R}}
\def\nat{\hbox{\rm\vrule\kern-1pt N}}

\def\l{\ell}

\newlength{\algowidth}
\newlength{\algowidthplus}
\newlength{\figwidth}
\begin{document}
\pagestyle{plain}

\begin{frontmatter}
\ifdefined\LONG
\title{Dynamic programming algorithms, efficient solution of the LP-relaxation and approximation schemes for the Penalized Knapsack Problem}
\else{}
\title{New exact approaches and approximation results for the Penalized Knapsack Problem}
\fi

%\corref{cor1}}
 %\cortext[cor1]{Corresponding author:}
 %\ead{}
\date{}

 \author[1,2]{Federico Della Croce}
 \author[3]{Ulrich Pferschy}
 \author[1]{Rosario Scatamacchia\corref{cor1}}
   %\ead[cor1]{rosario.scatamacchia@polito.it}
 \cortext[cor1]{Corresponding author: Rosario Scatamacchia}
 
  \address[1]{\small Dipartimento di Automatica e Informatica, Politecnico di Torino,\\ Corso Duca degli Abruzzi 24, 10129 Torino, Italy, \\{\tt \{federico.dellacroce, rosario.scatamacchia\}@polito.it }}
 \address[2]{CNR, IEIIT, Torino, Italy}
 \address[3]{\small Department of Statistics and Operations Research, University of Graz, Universitaetsstrasse 15, 8010 Graz, Austria, \\ {\tt pferschy@uni-graz.at}}

\begin{abstract}
We consider the 0--1 Penalized Knapsack Problem (PKP). Each item has a profit, a weight and a penalty and the goal is to maximize the sum of the profits minus the greatest penalty value of the items included in a solution. We propose an exact approach relying on a procedure which narrows the relevant range of penalties,
on the identification of a core problem and on dynamic programming. The proposed approach turns out to be very effective in solving hard instances of PKP and compares favorably both to commercial solver CPLEX 12.5 applied to the ILP formulation of the problem and to the best available exact algorithm in the literature. Then we present a general inapproximability result and investigate several relevant special cases which permit fully polynomial time approximation schemes (FPTASs). 
\end{abstract}


\end{frontmatter}

\section{Introduction}

We consider the 0--1 Penalized Knapsack Problem (PKP), as introduced in \cite{CeRi06}. 
PKP is a generalization of the classical 0--1 Knapsack Problem (KP) (\cite{MarTot90}, \cite{MarPisTot00}, \cite{KePfPi04}), where each item has a profit, a weight and a penalty. 
The problem calls for maximizing the sum of the profits minus the greatest penalty value of the items included in a solution. 

PKP has applications in resource allocation problems with a bi--objective function involving the maximization of the sum of profits against the minimization of the maximum value of a feature of interest. As an example, applications arise in batch production systems where the processing time/cost of batches of products depends on the maximum processing time of each product.
PKP also occurs as sub--problem within algorithmic frameworks designed for more complex problems. For instance, PKP arises as a pricing sub--problem in branch--and--price algorithms for the two--dimensional level strip packing problem in \cite{LoMaMO02}. 

PKP is $\mathcal{NP}$--hard in the weak sense since it contains the standard KP as special case, namely when the penalties of the items are equal to zero, and it can solved by a pseudo--polynomial algorithm. 
In \cite{CeRi06}, it is mentioned that a dynamic programming approach with time complexity $O(n^2c)$ can be easily determined, with $n$ and $c$ being the number of items and the capacity of the knapsack respectively. Also, an exact algorithm is presented and successfully tested on instances with $1000$ variables while running into difficulties on instances with 10000 variables. 
%The approach relies on solving standard knapsack problems induced by choosing 
%the item yielding the penalty value, denoted as \textit{leading item}, and discarding 
%the items with higher penalties. 
%The order of the sub--problems to be explored is greedily determined 
%according to the upper bounds given by their linear relaxation. 

In this work we propose an exact approach that relies on a procedure narrowing the relevant range of penalties and on dynamic programming.  
First, a straightforward pseudo--polynomial algorithm running with complexity $O(\max\{n\log n, nc\})$ is lined out. % lower than $O(n^2c)$ is lined out. 
Then, as our major contribution, we devise a dynamic programming algorithm based on a core problem and the algorithmic framework proposed in \cite{Pis97}. 
We investigate the effectiveness of our approach on a large set of instances generated according to the literature and involving different types of correlations between profits, weights and penalties. The proposed approach turns out to be very effective in solving hard instances and compares favorably to both solver CPLEX 12.5 and the exact algorithm in \cite{CeRi06}, successfully solving all instances with up to 10000 items.

Finally, we provide additional theoretical insights into the problem,
in particular about upper bounds. 
We derive a surprising negative approximation result, but also investigate several relevant special cases which still permit fully polynomial time approximation schemes 
(FPTAS). 
 
The paper is organized as follows. In Section \ref{sec:themodelPKP}, the linear programming formulation of the problem is introduced. In Section \ref{sec:PropPKP}, we provide some insights on structure and properties of PKP. We outline the proposed exact solution approach in Section \ref{PKP:ExactApp} and discuss the computational results in Section \ref{sec:CompResPKP}. 
Approximation results are presented in Section~\ref{PKP:APPROX}. 
Section~\ref{sec:PKPconclusion} concludes the paper with some remarks.
\ifdefined\LONG
Note that this report is a slightly extended version of a journal paper.
\else{}
Some of the proofs and additional computational results 
can be found in an accompanying technical report
\cite{longversion}
which is an extended version of this paper.
\fi{}

\section{Notation and problem formulation}
\label{sec:themodelPKP}

In PKP a set of $n$ items is given together with a knapsack with capacity $c$.
Each item $j$ has a weight $w_{j}$, a  profit $p_{j}$ and a penalty $\pi_j$. 
We will assume that all data are non--negative integer values.
Note that \cite{CeRi06} assume only the integrality of weights.
Similar to their work, all our algorithmic contributions 
do not require integral profits or penalties.
However, the theoretical approximation results of Section~\ref{PKP:APPROX}
partially rely on integral values.
We will assume that items are sorted in decreasing order of penalties, i.e.
\begin{equation}\label{eq:sort}
\pi_1 \geq \pi_2 \geq \ldots \geq \pi_n.
\end{equation}
The problem calls for maximizing the total profit minus the greatest penalty value of the selected items without exceeding the knapsack capacity $c$.
In order to derive an ILP-formulation, we associate with each item $j$ a binary variable $x_{j}$ such that $x_j=1$ iff item $j$ is placed into the knapsack.
Also, we associate a real variable $\Pi$ with the decrease in the objective produced by the highest penalty value of the items placed in the knapsack.
The problem can be formulated as follows:
\begin{eqnarray}
\mbox{(PKP)}\qquad \text{maximize} && \sum_{j=1}^n p_j x_j  - \Pi \label{eq:obj}\\
	\text{subject to}
	&& \sum_{j=1}^n w_j\, x_j\leq c\label{eq:weight}\\
	&& \pi_j\, x_j \leq \Pi \qquad j= 1,\dots,n \label{eq:penalty}\\
	&& x_j \in\{0,1\} \qquad j= 1,\dots,n \label{eq:binary}\\
	&& \Pi \in \real \label{eq:real}
\end{eqnarray}
(\ref{eq:weight}) is the standard capacity constraint.  
Constraints (\ref{eq:penalty}) ensure that $\Pi$ will carry the highest penalty value in any feasible solution of PKP;
variable $\Pi$ can be defined in (\ref{eq:real}) as real and will 
always attain one of the values $\pi_j$ in an optimal solution.
The objective function (\ref{eq:obj}) maximizes the sum of the profits minus the greatest penalty value of the selected items. 
Its optimal function value will be denoted by $z^*$.
For any considered sub--problem $PP$ the optimal objective function value will be written as $z(PP)$.
The item yielding the optimal penalty will be called the {\em leading item} and is denoted by $j^*$.

For further analysis, we will consider the penalty value $\Pi$ as a fixed parameter 
and define $PKP(\Pi)$ as an instance of PKP where the penalty value in the objective function is fixed to $\Pi$.
This simply means that all items $j$ with $\pi_j > \Pi$ are eliminated
from consideration and the remaining problem reduces to a standard 0--1 knapsack problem.
Obviously, $PKP(\Pi)$ only needs to be considered for the $n$ relevant choices $\Pi \in \{\pi_1,\ldots, \pi_n\}$.
This implies that PKP can be solved to optimality by solving at most $n$ KP instances of type $PKP(\Pi)$ and taking the maximum objective function value.
%Clearly, this approach is not expected to be effective as soon as the number of items and %thus of the penalty values increases.

\section{Generalities and algorithmic insights}
\label{sec:PropPKP}

We discuss here structure and properties of PKP. %We first show how to effectively compute upper bounds on the problem. 
We provide a characterization of the linear relaxation of PKP and propose an improved variant of the procedure proposed in \cite{CeRi06} for computing upper bounds on the sub--problems induced by the selection of the leading item. Then, we outline a basic dynamic programming algorithm with running time $O(\max\{n\log n, nc\})$. %lower than $O(n^2c)$. 
%(To discuss: could the linear relaxation part fit in the paper as follows?
%As we stated, the corresponding results can be used for computing an upper bound on PKP, as commonly is done in algorithmic frameworks. At the same time, the $O(n\log n)$ complexity may not sound appealing since we can compute better bounds with the same complexity. Also, we actually never use the LP relaxation bound in the algorithm.)

\medskip

\subsection{Computing upper bounds}
\label{PKP:UpperBounds}

%\subsection{Linear relaxation of PKP}
%\label{subsec:LP}

A natural upper bound on PKP is given by the optimal solution value of the linear relaxation of the problem, denoted as $PKP^{LP}$, where constraints (\ref{eq:binary}) are replaced by $x_j \in [0,1]$, i.e.\ items can be split and only a fractional part is packed. In this case a proportional part of the penalty applies. 
The optimal objective function value will be denoted by $z^{LP}$.
The LP-relaxation parametrically depending on $\Pi$ will be denoted as $PKP^{LP}(\Pi)$
with optimal solution value $z^{LP}(\Pi)$. 
%After sorting the items in decreasing order of efficiencies $p_j/w_j$, 
In the LP-relaxation for a given value $\Pi$, each variable $x_j$ is upper bounded by the following expression, since items with penalty exceeding $\Pi$ are reduced by scaling:
\begin{equation}\label{eq:fracbound}
x_j \leq \min\{1, \Pi/\pi_j\}.
\end{equation}
%This means that all items with a penalty exceeding the current $\Pi$, 
%are reduced by scaling
%such that their penalty contribution is equal to $\Pi$.
After imposing this bound the problem reduces again to an instance of KP (for fixed $\Pi$) 
for which the LP-relaxation is trivial.
%As usual, the split item is the first item that would exceed the capacity if set to its
%upper bound according to (\ref{eq:fracbound}).
We can give 
%Considering the optimal solution value of  $PKP^{LP}(\Pi)$ as a function in $\Pi$ we get
the following characterization for $PKP^{LP}(\Pi)$.

\begin{theorem}\label{th:concave}
$z^{LP}(\Pi)$ is a piecewise-linear concave function in $\Pi$ consisting of at most $2n$ linear segments.
\end{theorem}
\begin{proof}

Consider an arbitrary value of $\Pi$ and the corresponding solution 
$x^{LP}(\Pi)$.
Let $S$ denote the set of items $j$ with $x^{LP}_j(\Pi)=\Pi/\pi_j$,
i.e.\ all items whose values are currently bounded by the considered penalty value $\Pi$.
The current split item will be denoted as $s$.

We analyze the slope on the left-hand side of $(\Pi,z^{LP}(\Pi))$ by considering the change of the function
implied by a decrease of the penalty bound from $\Pi$ to $\Pi-\eps$
for some small $\eps>0$.
Formally, this change $\delta(\Pi)$ is given as follows:
\begin{eqnarray}
\delta(\Pi)=
z^{LP}(\Pi)-z^{LP}(\Pi-\eps) &=& -\,\eps 
+\underbrace{\sum_{j \in S} \frac{\eps}{\pi_j} p_j}_{\mbox{\tiny reduction of items in $S$}}
-\underbrace{\frac{p_s}{w_s} \sum_{j \in S}\frac{\eps}{\pi_j} w_j}_{\mbox{\tiny increase of split item}} \nonumber  \\
&=&\eps\left(-1+\sum_{j \in S} \frac{w_j}{\pi_j}   \underbrace{\left(\frac{p_j}{w_j}-\frac{p_s}{w_s}\right)}_{\geq 0}  \right) 
\label{eq:slope}
\end{eqnarray}
This expression can be positive or negative, but it shows that $\delta(\Pi)$ is proportional to $\eps$.
Thus, in a neighborhood of $(\Pi, z^{LP}(\Pi))$ the function consists of a linear piece which will end in one of the following three cases:
\begin{enumerate}
\item $\Pi-\eps = \pi_k$ for some $k \not\in S$.
This means that lowering $\Pi$, a new item is found for inclusion in $S$.
Plugging in the extended set $S$ in (\ref{eq:slope}) will clearly increase the change $\delta(\Pi)$.
\item $x_s$ reaches $1$.
This means that the split item becomes integral and item $s+1$ becomes the new split item.
Thus, we replace $\frac{p_s}{w_s}$ by $\frac{p_{s+1}}{w_{s+1}}$ in (\ref{eq:slope}) again implying an increase of $\delta(\Pi)$.
\item $x_s$ reaches $(\Pi-\eps)/\pi_s$. This means that the increase of the split item reaches the lowered penalty bound.
In this case, $s$ is included in $S$ and $s+1$ becomes the new split item.
Again, $\delta(\Pi)$ is increased by combining both of the above arguments.
\end{enumerate}
Summarizing, we have shown that starting with an arbitrary value of $\Pi$ and 
decreasing $\Pi$, $z^{LP}(\Pi)$ consists of a linear piece which ends with some $\Pi'\leq \Pi$ 
in one of three possible configurations.
The slope of this linear piece is given by $\frac{\delta(\Pi)}{\eps}$.
The preceding linear segment of $z^{LP}(\Pi')$ will have an {\em increased} change $\delta(\Pi')$
if $\Pi'$ is further decreased. 
This means that the previous linear segment at $z^{LP}(\Pi')$ has a {\em larger slope} than
$z^{LP}(\Pi)$.
Thus, the slope of $z^{LP}(\Pi)$ is {\em decreasing} with increasing $\Pi$
which proves the concavity of $z^{LP}(\Pi)$.

Starting the above procedure with $\Pi=\max_{j=1}^n \pi_j$ and reducing $\Pi$ iteratively until $\Pi=0$,
it is clear that each item may cause the end of a linear segment of $z^{LP}(\Pi)$ at most twice:
Once, by becoming a new split item and once by being included in set $S$.
Each such event can occur at most once for each item.
Therefore, there can be at most $2n$ linear pieces.
\end{proof}

Exploiting this characterization one can easily construct a solution algorithm based on binary search over the penalty space with $O(n\log \pi_1)$ time.\\
\ifdefined\LONG
Some more effort is necessary to reach a binary search only over the $n$ relevant penalty values, which yields the following proposition.

\begin{prop}\label{th:PKPLPnlogn}
$PKP^{LP}$ can be solved in $O(n\log n)$ time.
\end{prop}

%\ifdefined\LONG

\begin{proof}
Algorithmically, one can easily exploit the structure established in Theorem~\ref{th:concave} by performing a binary search over $\Pi$ to determine a maximum\footnote{Note that the maximum  is not necessarily unique, since there may exist a linear segment of $z^{LP}(\Pi)$ with slope $0$.} of the concave function $z^{LP}(\Pi)$.
For each query value $\Pi$, one can compute the split item in linear time and also assemble the corresponding set $S$ in one pass through the set of items. Thus, for each query value $\Pi$ the sign of the slope can be calculated from (\ref{eq:slope}) in linear time.\\
Applying the binary search over all possible penalty values would yield a total running time of $O(n \log \pi_1)$ which is polynomial in the size of the (binary) encoded input, i.e.\ weakly polynomial. To obtain a strongly polynomial time algorithm whose running time depends only on the number of input values, we can first perform a binary search over all $n$ values $\pi_j$ and thus compute in $O(n\log n)$ time the interval of two consecutive penalties $[\pi_k,\pi_{k-1}]$ for some $k$ (with $\pi_{k-1} \geq \pi_k$ according to (\ref{eq:sort})) which will include the optimal penalty value $\Pi^\LP$. \\ 
Let us denote the optimal split item by  $s^{\LP}$ and the split items associated with penalties $\pi_{k-1}$ and $\pi_k$ by $s_{k-1}$ and $s_{k}$ respectively. If we consider the item sorted by decreasing $\frac{p_j}{w_j}$, item $s_{k-1}$ will precede item $s_{k}$ in the ordering.\\
If $\pi_{k-1} = \pi_k$, clearly we have $s^\LP=s_{k-1}=s_{k}$ and $\Pi^{LP}=\pi_{k-1}=\pi_k$. Otherwise, we have to find $s^\LP$ in the interval $[s_{k-1},s_k]$ and related $\Pi^{LP}$ in the interval $[\pi_k,\pi_{k-1}]$. \\
Let us consider a generic item $\overline{s}$ as candidate for the split item. The best penalty $\Pi$ associated with it can be calculated as follows. Given the interval $[\pi_k,\pi_{k-1}]$, we set $x_j = 1$ ($j=1, \dots, \overline{s} - 1$) if $\pi_j \leq \pi_k$, $x_j = \frac{\Pi}{\pi_j }$ otherwise. 
This implies that the weight and profit sums of the items preceding $\overline{s}$ are linear functions of $\Pi$ in the form
\begin{equation}
\label{wsumS-1)}
\sum_{j=1}^{\overline{s}-1} w_j x_j =  \gamma_1 \Pi + \gamma_2,
\end{equation}
\begin{equation}
\label{psumS-1)}
\sum_{j=1}^{\overline{s}-1} p_j x_j =  \theta_1 \Pi + \theta_2,
\end{equation}
with non-negative coefficients $\gamma_1$, $\gamma_2$, $\theta_1$, $\theta_2$ determined according to the above $x_j$ setting.
Item $\overline{s}$ can be the split item if and only if $\sum_{j=1}^{\overline{s}-1} w_j x_j < c$ and its value $x_{\overline{s}}$ fulfills the capacity (i.e. $x_{\overline{s}}=\frac{c -\sum_{j=1}^{\overline{s}-1} w_j x_j}{w_{\overline{s}}}$) while satisfying the bound (\ref{eq:fracbound}). 
Correspondingly, the feasible interval of $\Pi$, denoted as $I_{\overline{s}}(\Pi)$, which allows $\overline{s}$ to be the split item is defined by the following system of inequalities:

\begin{equation}
\label{SystemI(Pi)}
I_{\overline{s}}(\Pi) := \begin{cases} \pi_k \leq \Pi \leq \pi_{k-1}\\
 \gamma_1 \Pi + \gamma_2  \leq c - \eps \\ 
\frac{c - \gamma_1 \Pi - \gamma_2}{w_{\overline{s}}} \leq \beta \quad \mbox{ with } \beta =
\begin{cases}
1 \mbox{  }\mbox{ if $\pi_{\overline{s}} \leq \pi_k$;}\\
\frac{\Pi}{\pi_{\overline{s}}} \mbox{ otherwise}
\end{cases}
\end{cases}
\end{equation} 
Item $\bar{s}$ is a relevant candidate for the split item only if the corresponding interval $I_{\overline{s}}(\Pi)$ is non-empty. In such a case, the overall profit given by $\overline{s}$ as split item is 
\begin{equation}
\label{psumS}
P_{\overline{s}}(\Pi) = \sum_{j=1}^{\overline{s}-1} p_j x_j + p_{\overline{s}}x_{\overline{s}} - \Pi = (\theta_1 - \frac{p_{\overline{s}}}{w_{\overline{s}}}\gamma_1 - 1)\Pi + \theta_2 + \frac{p_{\overline{s}}}{w_{\overline{s}}}(c - \gamma_2)
\end{equation}
and will be maximized by choosing either the left extreme of $I_{\overline{s}}(\Pi)$ if the term $(\theta_1 - \frac{p_{\overline{s}}}{w_{\overline{s}}}\gamma_1 - 1) < 0$ or the right extreme otherwise.\\
Summarizing, the best penalty value associated with a candidate item can be computed in constant time if coefficients $\gamma_1$, $\gamma_2$, $\theta_1$, $\theta_2$ are given.
Hence, we may first consider item $s_{k-1}$ as candidate for $s^\LP$ and compute related coefficients in (\ref{wsumS-1)})--(\ref{psumS-1)}) and penalty value. Then, we iteratively move to the next item after updating coefficients $\gamma_1$, $\gamma_2$, $\theta_1$, $\theta_2$ in one pass due the inclusion of the previous candidate item among items $j=1,\dots,\overline{s}-1$. 
After the evaluation of item $s_{k}$, the optimal split item $s^\LP$ and penalty $\Pi^\LP$ are returned.
Since the execution time of this part is bounded by $O(n)$, the overall complexity for solving the LP relaxation is $O(n \log n)$. 
\end{proof}

In fact, we can do even better by interleaving a median search for the optimal penalty value with a median search for the split item.

\begin{prop}\label{th:PKPLPnloglogn}
$PKP^{LP}$ can be solved in $O(n\log\log n)$ time.
\end{prop}

\begin{proof}
We describe an algorithm performing an iterative median search for $\Pi^{LP}$.
In each step of the algorithm we are given a feasible interval for $\Pi^{LP}$
denoted as $[\eta^\l, \eta^u]$.
Then we determine the median $\eta^m$ among all penalty values $\pi_j$ in this interval.
This value $\eta^m$ is considered as a candidate for $\Pi^{LP}$.
%The structure of the solution corresponding to such a penalty candidate was described in Section~\ref{subsec:LP}:Items are sorted in nonincreasing order of efficiency and packed into the knapsack by a greedy strategy. Items $j$ with penalty value $\pi_j \leq \eta^m$ are fully packed, items with $\pi_j > \eta^m$ are truncated and packed with a weight $w_j \frac{\eta^m}{\pi_j}$. This greedy packing continues until the capacity $c$ is reached. A fractional part of the final split item is used to fill the capacity completely.
For the resulting solution it is easy to determine the slope of $z^{LP}(\eta^m)$
and decide accordingly whether $\eta^m$ is set as new upper or lower bound of the search interval.

To find the split item more efficiently in each iteration
we will avoid sorting the items by efficiency but instead employ a linear time
algorithm as described in \cite[ch.~3.1]{KePfPi04}.
It is based on iterated bisection of the item set into sets of items with
higher resp.\ lower efficiency than a current target value $e_i$.
We will keep a sorted list $E$ consisting of all target  values $e_i$ 
(in decreasing order of efficiencies) 
considered so far during the search for a split item over the different penalty candidate values.

Then we introduce the following data structures.
Set $S_i$ contains all items with a penalty value in the current search interval
and efficiency value between $e_i$ and the next larger efficiency target value 
$e_{i-1}$ in $E$. 
Formally, 
$$S_i :=\{ j\in N \mid \eta^\l \leq \pi_j \leq \eta^u, e_{i-1} > e_j \geq e_i\}.$$
For each set $S_i$ we also define the canonical weight sum
$W(S_i) := \sum_{j \in S_i} w_j$.

For the currently considered interval bounds $\eta$ we use the following auxiliary
weight arrays for each $e_i \in E$:
\begin{eqnarray}
W_i(\eta) &:=& 
\sum_{j}w_j \mbox{ over all } j \mbox{ with } e_j > e_i \mbox{ and } \pi_j \leq \eta\label{eq:aux1}\\
R_i(\eta) &:=& 
\sum_{j}w_j \mbox{ over all } j \mbox{ with } e_j > e_i \mbox{ and } \pi_j > \eta \label{eq:aux2}
\end{eqnarray}
These arrays allow an easy evaluation of each target value $e_i \in E$
since the total weight of all items with efficiency higher than $e_i$
and reduced to the current penalty bound $\eta$ is given by
$W_i(\eta)+R_i(\eta)$. 
Taking also the weight of the item with efficiency $e_i$ into account,
one can easily determine in constant time for a give target $e_i$ whether 
the split item is found or a higher resp.\ lower target value for the efficiency
should be considered.

It remains to explain the update of the auxiliary data structures.
Whenever a new efficiency target value $e_k$ is generated for some item $k$
with $e_{i-1} \geq e_k \geq e_i$ and inserted into $E$, 
the current set $S_i$ is partitioned into two sets, namely $S_k$ and (a new set) $S_i$
with the obvious definition.
Also the array entries $W_k(\eta)$ resp.\ $R_k(\eta)$ are generated
from $W_{i-1}(\eta)$ resp.\ $R_{i-1}(\eta)$, while
$W_i(\eta)$ resp.\ $R_i(\eta)$ remain unchanged.
These update operations can be done trivially by considering all items of the
original set $S_i$ explicitly.

In each main iteration, a new penalty search value $\eta^m$ is considered.
Therefore, all sets $S_i$ are bipartitioned into two disjoint sets $S_i^\l$ and $S_i^u$
with $S_i^\l \cup S_i^u =S_i$ and $\pi_j \leq \eta^m$ for $j \in S_i^\l$,
resp.\ $\pi_j > \eta^m$ for $j \in S_i^u$.
After deciding on the new search interval one of these two sets will 
replace $S_i$.
Furthermore, the auxiliary array entries are generated for $\eta^m$.
This can be done by setting
\begin{eqnarray}
W_i(\eta^m) &:=& W_i(\eta^\l) + \sum_{k\leq i} W(S_k^\l),\label{eq:update1}\\
R_i(\eta^m) &:=& R_i(\eta^u) + \sum_{k\leq i \atop j \in S_k^u} w_j \frac{\eta^m}{\pi_j}\,.
\label{eq:update2}
\end{eqnarray}

Concerning the running time, the binary search over all $n$ penality values, 
i.e.\ over all candidates $\eta^m$,
requires $O(\log n)$ iterations.
In each iteration we consider explicitly each item with a penalty value in the current interval $[\eta^\l, \eta^u]$ to determine the median $\eta^m$ and to update
the auxiliary arrays in (\ref{eq:update1}) and (\ref{eq:update2}).
(Note that these array entries also have to be generated if no item in the current
interval is involved.
However, since $|E| \leq \log^2 n$ there can be at most $\log^3 n$ such entries.)
Since the search interval is bisected in each iteration, this sums up to $O(n)$ time.

The effort for finding the split items (and thus the solution of the linear relaxation
of the KP implied by the current penalty bound $\eta$)
depends on the sequence of efficiency target values considered in each iteration.
Here, we will exploit the fact that evaluating a target value can be done in constant time 
employing (\ref{eq:aux1}) and (\ref{eq:aux2}). 
If the target value was already considered in an earlier iteration for a previous penalty
search value, then it is included in $E$ and no additional effort is required.
Note that at most $O(\log^2 n)$ such evaluations take place during the execution of the algorithm.
If the target value, say $e_k$, is considered for the first time,
the corresponding set $S_k$ has to be generated as described above from
a previously existing set $S_i$, which requires considering all items in $S_i$.

In the first iteration (for the median of all penalty values),
searching for the split item starts with $O(n)$ time for the first median (over all efficiencies), then another $O(n/2)$ time for the second target value 
(i.e.\ the median of the upper or lower half of efficiency values), and so on.
Searching for $t$-th target value will require $O(n/2^{t-1})$ time.
This includes also going through all items of the associated set $S_i$
which is of cardinality $n/2^{t-1}$.

The same holds for the second iteration, except for the first target value,
since the median of all efficiencies was for sure considered in the first iteration.
The second target value may or may not have been considered in the first iteration. 
Thus, we have to take the corresponding effort of $O(n/2)$ time into account.

In the third iteration, the effort for the second target value is only relevant,
if it was not considered in the second iteration.
Generalizing this argument over all iterations and taking -- for the time being -- 
only the first $t$ efficiency
target values into account, it turns out that
the effort for deciding the $t$-th target values in total over all $\log n$ iterations 
can be at most $O(n)$.
This results from considering each of the subsets of $n/2^{t-1}$ items at most once.
The total effort of this part is $O(n\cdot t)$.

Continuing the analysis for target values numbered by $t+1, t+2, \ldots, \log n$
we can bound the effort for each iteration over one penalty value
by 
$$n/2^t + n/2^{t+1} + \ldots + 1 \approx n/2^{t-1}.$$
This effort arises for all $\log n$ iterations over all penalty search values.
Thus, we can summarize the total running time associated to the solutions of the linear relaxations as:
$$n \cdot t + \log n \frac{n}{2^{t-1}}$$
Plugging in $t=\log\log n$ yields
$O(n \log\log n + n)$.
\end{proof}
\else{}
More effort is required to compute $PKP^{LP}$ in $O(n\log n)$ time. In fact, an even stronger result holds as $PKP^{LP}$ can be solved in $O(n\log\log n)$ time.
Since we will not apply the solution of $PKP^{LP}$ in our algorithms,
we defer the proof of these running times results to the accompanying technical report \cite{longversion}.
\fi{}

%Although $z^{LP}$ can be computed efficiently, we note however that the difference between $z^{LP}$ and the integer optimal solution value can be arbitrarily large.
%\begin{prop}\label{th:largegap}
%There are instances of PKP where $z^{LP}-z^*$ is arbitrarily large.
%\end{prop}
%\begin{proof}
%Consider the following PKP instance with $n$ items, capacity $c=W$ and the following entries: $w_j=W$, $p_j=M$ and $\pi_j = M -1$ for $j=1,\ldots,n$, with $M$ being an arbitrary large  number.
%Any optimal solution of this problem can pack only one item $j$ with optimal solution value $z^*=p_j - \pi_j=1$.
%The linear relaxation sets $x_j=\frac{1}{n}$ for $j=1,\ldots,n$ and $\Pi = \frac{1}{n}(M-1)$ yielding a solution value % $$ z^{LP} = n\frac{M}{n} - \frac{1}{n}(M-1)=
%\left(1 - \frac{1}{n}\right)M + \frac{1}{n}$$
%Hence, by choosing large values of $M$, the value of $z^{LP}$ can become arbitrarily large compared to the optimal value $z^*=1$.
%\end{proof}

%\subsection{Computing upper bounds}
%\label{PKP:UpperBounds}

\medskip 

We may compute more involved upper bounds on PKP as follows. 
As pointed out above, the optimal solution of PKP is determined by a penalty value $\Pi$ and a subset of items $j$ with $\pi_j \leq \Pi$.
Therefore, we consider sub--problem $PKP_j:=PKP(\pi_j)$ for $j=1,\ldots,n$. 
Recalling (\ref{eq:sort}) each $PKP_j$ is an instance of KP with item set 
$\{j, j+1, \ldots, n\}$ and capacity $c$
where $\pi_j$ is subtracted from the final solution value.

Fixing $\Pi=\pi_j$ for some $j$ is only relevant for the final solution if item $j$ is actually included in the solution.
Hence, as in \cite{CeRi06}, we also consider sub--problem $PKP_{j}^+$,
where item $j$ is packed, a fixed penalty of $\Pi=\pi_j$ is subtracted from the objective function,
and for the reminder of the solution a KP is solved with capacity $c-w_j$ and item set 
$\{j+1, \ldots, n\}$.
%$\{i\in \{1,\ldots,n\}\mid \pi_i \leq \pi_j\}$.

%$PKP_{j}$ and $PKP_{j}^+$ are trivial to solve if the sum of all relevant item weights does not exceed the available capacity.
For both $PKP_j$ resp.\ $PKP_j^+$ we consider the LP-relaxation as upper bound denoted by $PKP_j^{LP}$ resp.\ $PKP_j^{+LP}$.
It is easy to see that 
\begin{align}
&z(PKP_j^+) \leq z(PKP_j) \label{UBPPKP0}\\
&z^* = \max_{j=1,\ldots,n} z(PKP_j) \leq \max_{j=1,\ldots,n} z(PKP_j^{LP}) =: UB_{sub} \label{UBPKP1}\\
&z^* = \max_{j=1,\ldots,n} z(PKP_j^+) \leq \max_{j=1,\ldots,n} z(PKP_j^{+LP}) =: UB_{sub}^+ \label{UBPKP2}
\end{align}
The following dominance relations exist for the upper bounds $UB_{sub}$, $UB_{sub}^+$ and $z^{LP}$. 

\begin{prop}
For any PKP instance, we have that
\begin{equation}
\label{RelBounds}
UB_{sub}^+  \leq UB_{sub} \leq z^{LP}
\end{equation}
and there are instances where the inequalities are strict.
\end{prop}
\begin{proof}
Clearly, the restricted feasible domain of $UB_{sub}^+$ cannot lead to a greater value than $UB_{sub}$ and thus $UB_{sub}^+  \leq UB_{sub}$.
%To prove that $UB_{sub} \leq z^{LP}$,
Let us denote by $j^\prime$ the item yielding $UB_{sub}$, i.e.\ $UB_{sub} = z(PKP_{j^\prime}^{LP})$. Computing $PKP^{LP}(\Pi)$ with $\Pi=\pi_{j^\prime}$ gives a feasible solution for the LP relaxation whose value is less than (or equal to) the optimal value $z^{LP}$ but at least as large as $UB_{sub}$. The latter holds because {\em all} items are involved (and bounded according to (\ref{eq:fracbound})) in the computation while only items $i$ with $\pi_i \leq \pi_{j^\prime}$ are considered for solving $PKP_{j^\prime}^{LP}$. This implies that $UB_{sub} \leq z^{LP}$. 

To show that inequalities in (\ref{RelBounds}) can be strict, consider the following PKP instance with $n=2$ items, capacity $c=7$ and the entries:
$$ p_1= 10, w_1 = 5, \pi_1 = 1;\quad p_2= 6, w_2 = 4, \pi_2 = 2$$
%\begin{table}[h]
%	\centering
	%\scriptsize
%\begin{tabular}{|r|rr|}
 % \hline
%$j$ & 1 & 2\\ \hline
%$p_j$ & 10 & 6\\
%$w_j$ & 5 & 4\\
%$\pi_j$ & 1 & 2\\ \hline
%\end{tabular}
%\end{table}
%This may occur if the item $j$ implying the penalty value of the maximum expression turns out to be a split item of the corresponding LP-solution.
For this instance we have $z^{LP} = 12$, $z(PKP_1^{LP}) = -1 + 10 =9$,
$z(PKP_2^{LP}) = -2 + 10 + \frac 2 4 6= 11$, $z(PKP_1^{+LP}) = -1 +10 =9$, $z(PKP_2^{+LP}) = -2 + 6 + \frac 3 5 10 =10$. Thus, we have:
$$ UB_{sub}^+ = 10 < UB_{sub} = 11 < z^{LP} = 12$$
\end{proof}

Although the three bounds can be computed efficiently and can be expected to be reasonably close to the optimal value in practice, 
we show a negative result on their deviation from the optimum.
%in analogy to Proposition~\ref{th:largegap}.
\begin{prop}
There are instances of PKP where the differences ($UB_{sub}^+ -z^*$), ($UB_{sub} -z^*$) and ($z^{LP} -z^*$) are arbitrarily large.
\end{prop}

\begin{proof}
Consider the following instance with $n=2$ items, capacity $c= M$ and the following entries: $p_j=w_j=\frac{M}{2} + 1$ and $\pi_j=\frac{M}{2}$ for $j=1,2$.
In an optimal solution only one item $j$ is packed and, correspondingly, $z^*= p_j - \pi_j= 1$. 
Also, it is easy to see that $UB_{sub}^+ = \frac{M}{2}$, which, in combination with (\ref{RelBounds}), shows the claim.
%Thus, by setting $M$ to large values, both bounds $UB_{sub}$ and $UB_{sub}^+$ can be arbitrarily greater than $z^*$.
\end{proof}

Algorithmically, it is not difficult to see that all values $z(PKP_j^{LP})$ for $j=1,\dots,n$ can be computed in $O(n \log n)$ time. 
Also from a practical point of view, the effort hardly exceeds sorting.
As a preprocessing step an auxiliary array is constructed containing all items sorted in decreasing order of efficiencies $p_j/w_j$.
Then the problems $PKP_j^{LP}$ are considered iteratively for $j=1,\ldots, n$,
i.e., in decreasing order of penalties $\pi_j$.
First, $PKP_1^{LP}$ is solved in linear time and the corresponding split item 
(i.e.\ the first item not fully packed into the knapsack) is identified.
We keep a pointer to this split item in the sorted array of items.
Moving to $PKP_2^{LP}$, we just remove item $1$ from the solution and increase the split item,
or possibly move to a new split item by shifting the pointer towards items with lower efficiency.
All together, after sorting, all values $z(PKP_j^{LP})$ can be determined in linear time
by one pass through the sorted array of items.

\medskip

In \cite{CeRi06}, the authors presented an $O(n^2)$ procedure to compute all values $z(PKP_j^{+LP})$ for $j=1,\dots, n$. In the following, we show that in fact $O(n\log n)$ time is sufficient to perform this task.

\begin{prop}\label{th:nlogn}
All values $z(PKP_j^{+LP})$ for $j=1,\ldots, n$ can be computed in $O(n\log n)$ time.
\end{prop}

\begin{proof}
First, the items are sorted in decreasing order of efficiencies.
Based on this sequence we construct 
an auxiliary data structure consisting of a binary tree as follows:
Each item corresponds to a leaf node of the tree.
These are nodes at level $0$.
A parent node is associated with each pair of consecutive items (with a singleton remaining at the end for $n$ odd) thus yielding other $\lceil \frac n 2 \rceil$ nodes in the level $1$ of the tree.
This process is iterated recursively up the tree, which trivially reaches a height of $O(\log n)$.

In each node $v$ of the tree we store as $W(v)$ (resp.\ $P(v)$) the sum of weights (resp.\ profits) of all items corresponding to leaf nodes in the subtree rooted in $v$.
Clearly, such a tree and its additional information can be built in $O(n)$ time.

For any given capacity $c'$ the corresponding split item and also the value of the optimal LP-relaxation can be found in $O(\log n)$ time by starting at the root node and going down towards a leaf node by applying the following rule in every node $v$ with left and right child nodes $v^L$ and $v^R$:
\begin{quote}
If $W(v^L) > c'$ then set $v:=v^L$.\\
Otherwise set $v:=v^R$ and $c':=c'-W(v^L)$.
\end{quote}
The item corresponding to the leaf node reached by this procedure is the split item. 
%of which an appropriate part is included in the solution.
The solution value can be reported by keeping track of the $P(v)$ values during the pass through the tree.

In the main iteration of the algorithm we compute $z(PKP_j^{+LP})$ iteratively for $j=1,\ldots, n$ in decreasing order of penalties $\pi_j$.
First we remove item $j$ permanently from consideration.This means that the leaf node corresponding to $j$ is removed from the tree and
all $O(\log n)$ labels $W(v)$ (resp.\ $P(v)$) on the unique path from this leaf to the root of the tree are updated by subtracting $w_j$ (resp.\ $p_j$).
Then we solve an LP-relaxation with capacity $c':=c-w_j$ and add $p_j-\pi_j$ to the objective function. All together there are $n$ iterations, each of which requiring $O(\log n)$  time to find the solution of the LP-relaxation and $O(\log n)$ time to update the 
labels of the binary tree. 
\end{proof}
Note that we might expect a considerable speed--up of the running time $O(n\log n)$ in a practical implementation since the tree looses vertices in each iteration and path contractions can be performed.

\subsection{A basic dynamic programming algorithm}
\label{PKP:BasicDp}

As recalled in \cite{CeRi06}, a straightforward pseudo--polynomial algorithm for PKP consists of solving $j$ standard knapsack problems $PKP_j^+$ by the classical dynamic programming by weights running in $O(nc)$. The overall complexity is thus $O(n^2c)$. 
However, we can do much better by evaluating all $n$ subproblems in one run.
%The following proposition shows that the complexity can be reduced to $O(\max\{n\log n, nc\})$.
\begin{theorem}
\label{TrivialDP}
PKP can be solved with complexity $O(\max\{nc, n\log n\})$.
\end{theorem}
\begin{proof}
It suffices to consider the items sorted by increasing penalty and to run the dynamic program for KP only once. If we denote by $F_j(d)$ the optimal solution value of the sub--problem of KP consisting of items ${1,\dots, j}$ and capacity $d \leq c$, the optimal value of any PKP instance is simply given by
 \begin{equation}
\underset{j=0,\dots,n-1}{\operatorname{\,max}}\; \{F_j(c - w_{j+1}) + p_{j+1} - \pi_{j+1} \}
\end{equation}
That is, we evaluate the choice of item $j+1$ as leading item just by considering the maximum profit 
reachable with the previous items in a knapsack with capacity $c - w_{j+1}$. 
The running time is $O(nc)$ plus the effort for sorting.
\end{proof}

\section{An exact solution approach}
\label{PKP:ExactApp}

\subsection{Overview}
\label{PKP:Rationale}

The DP algorithm of Theorem \ref{TrivialDP}, hereafter denoted as $DP_1$, may be appealing whenever the capacity $c$ is of reasonably limited size. 
However, for KP it is known that more effective than the iterative addition of items
are algorithms based on the core problem.
Thus, our idea is to exploit the core concept for PKP similarly to the framework of the \textit{Minknap} algorithm \cite{Pis97}. % which we briefly recall 
We remark that the presence of penalties compromises in PKP the structure of an optimal solution with respect to a standard KP. 
This difference would typically affect the performance of an approach based on a core problem. 
Further, the presence of penalties limits the effectiveness of the classical dominance rule in KP based on the profits and the weights of the states. 
Anyhow, from a practical perspective it is still beneficial to run a dynamic programming algorithm 
starting from the split solution of KP and not from scratch. 
In addition, by narrowing the interval of penalty values which can possibly lead to an optimal solution, the ``noise'' added by the penalties can be further reduced.

We propose an exact approach involving two main steps. In the first step, we effectively compute an initial feasible solution for the problem and identify the relevant interval of penalties values possibly leading to an optimal solution. In the second step, we run a dynamic programming algorithm with states based on the core concept. In case the first step yields a reduced problem with a reasonably limited input size, we could as well launch the $DP_1$ algorithm. 
In the following subsections we describe the steps of the approach whose pseudo code is presented in Algorithm \ref{algo:EAPKP}. 
\begin{algorithm}
\small
	\caption{Exact solution approach}
	\label{algo:EAPKP}
	\begin{algorithmic}[1]
	\State{\textbf{Input:} PKP instance, parameters $T_1$, $T_2$, $T_3$, $\alpha$.}
     
     \Comment{Step 1}

      \State{$KP_1$ = PKP without penalties; } 
	 % \State{$(\overline{z}, \overline{j}, \overline{f}) \leftarrow$ Solve $KP_1$ by \textit{ModMinknap};}
	  \State{$(\overline{z}, \overline{j}, \overline{f}) \leftarrow$ \textit{ModMinknap($KP_1$)};}
	  \medskip
	 % \State{Discard items $j = 1,\dots,\overline{f}$;}
	  %\State{Fix variables $x_j = 0$ for $j=1,\dots,\overline{f}$;}
		  
 \State{Compute $z(PKP_j^{+LP})$ for $j=\overline{f}+1,\ldots,n$};
 \State{$UB = \underset{j}{\operatorname{max}}\; z(PKP_j^{+LP})$;}
	  
\medskip
	    \IiIf{$UB \leq \overline{z}$}
	     $z^* = \overline{z}$, $j^* = \overline{j}$; \Return $(z^*,j^*)$;		
		 \EndIiIf 	
\medskip

 \State{$k=\underset{j}{\operatorname{arg\,max}}\;z(PKP_j^{+LP})$;}
		 \State{$KP_2$ = $KP_1$ $\cap$ ($x_j =0$ $j=1,\dots,k-1$);} 
	    %\State{$(\hat{z}, \hat{j}, \hat{f}) \leftarrow$ Solve $KP_2$ by \textit{ModMinknap};}	  
         \State{$(\hat{z}, \hat{j}, \hat{f}) \leftarrow$ \textit{ModMinknap($KP_2$)};}	
 \medskip
		  \IiIf{$\hat{z} > \overline{z}$}
		  $\overline{z} = \hat{z}$, $\overline{j} = \hat{j}$;\EndIiIf 
		  %\If{$\hat{z} > \overline{z}$}
		  %\State {$\overline{z} = \hat{z}$, $\overline{j} = \hat{j}$};		
		  %\EndIf 	
 \medskip
 		\State{$l \leftarrow$ Apply (\ref{eq:Compleft});}
		% \State{$l = \underset{}{\operatorname{\,min}}\; \{j :  z(PKP_j^{+LP}) > \overline{z} \}$};  
		\State{$r \leftarrow$ Apply (\ref{eq:Compright});}
		% \State{$r = \underset{}{\operatorname{\,max}}\; \{j : z(PKP_j^{+LP}) > \overline{z} \}$};
		 \State{$\Pi_{max}=\pi_l$};
		 \State{$\Pi_{min}=\pi_r$};
	
		 %\State{Discard items $j=\overline{f}+1,\dots,l - 1$;}
		  %\State{Fix variables $x_j = 0$ for $j=\overline{f}+1,\dots,l - 1$;}
  
  \medskip
  
   \IiIf{$[\Pi_{min}, \Pi_{max}] = \emptyset$}
	     $z^* = \overline{z}$, $j^* = \overline{j}$; \Return $(z^*,j^*)$;		
		 \EndIiIf 	
  
  \medskip
    
     \State{$PKP'$ = PKP $\cap$ ($x_j =0$ $j=1,\dots,l - 1$; $\Pi \geq \Pi_{min}$)}; 
    \State{$n'=n-l+1$;}   %without items $j=1,\dots,l - 1$; $n'=n-l+1$;} 
           
  \medskip
       
        \Comment{Step 2}	
      
	       \If{$n'c \leq T_1$ and $(r-l+ 1) \geq T_2$}
	       % \State {$(z',j') \leftarrow$ Solve $PKP'$ by $DP_1$};
	        \State {$(z',j') \leftarrow$  $DP_1(PKP')$};
	         \IfThenElse{$z' > \overline{z}$} 
		      {$z^* = z'$, $j^* = j'$;}
		      {$z^* = \overline{z}$, $j^* = \overline{j}$};
             \EndIfThenElse
		    \Else 
		   % \State {$(z',j') \leftarrow$ Solve $PKP'$ by $DP_2$};	
		   %  \State {$(z',j') \leftarrow$ $DP_2(PKP',\Pi_{min},\overline{z},\overline{j},T_3,\alpha)$};	
		   	  \State {$(z^*,j^*) \leftarrow$ $DP_2(PKP',\overline{z},\overline{j},T_3,\alpha)$};	
		 \EndIf 	
  \medskip
											
	\State \Return $(z^*,j^*)$;
\end{algorithmic}
\end{algorithm}

\subsection{Step 1: Computing an initial feasible solution and the relevant interval of penalty values}
\label{PKP:Rationale}

The approach takes as input four parameters $T_1$, $T_2$, $T_3$, $\alpha$ and starts by solving the standard knapsack problem $KP_1$ given by disregarding the penalties of the items in PKP (lines 2-3 in Algorithm \ref{algo:EAPKP}).
This problem is solved as follows. Denote the index of the first item in the optimal solution of $KP_1$ (according to the ordering (\ref{eq:sort})) by $\overline{f}$. The corresponding first feasible solution of PKP has objective value $z(KP_1) - \pi_{\overline{f}}$.
Similarly to Proposition 2 in \cite{CeRi06}, the following proposition holds
\begin{prop}\label{th:FromKPprop}
All items $j = 1, \dots, \overline{f}-1$ can be discarded without loss of optimality.  
\end{prop}

\begin{proof}
Since $z(KP_1)$ is the optimal solution value, including any item $j = 1, \dots, \overline{f}-1$ leads to a solution with profits less than (or equal to) $z(KP_1)$ but induces a penalty greater than (or equal to) $\pi_{\overline{f}}$. 
%Consequently, every solution of PKP with at least one of these items included cannot improve the first solution value $z(KP_1) - \pi_{\overline{f}}$.
\end{proof}

Thus, if there is more than one optimal solution of $KP_1$, we are interested in the solution yielding the lowest penalty value for PKP, i.e.\ the largest index $\overline{f}$. 
This task is easily accomplished by considering a slight variant of \textit{Minknap}, hereafter denoted as \textit{ModMinknap}, which keeps track of all optimal solutions of $KP_1$ and the corresponding penalty values in PKP. 
In addition, we can compute PKP solutions during the iterations of \textit{ModMinknap} just by tracking the largest penalty associated to each feasible state. 
We then take the overall best solution found for PKP. Denote by $\overline{z}$ its value and by $\overline{j}$ the index of the leading item.

We remark that %\textit{ModMinknap} this variant of \textit{Minknap} 
\textit{ModMinknap}
is only a heuristic algorithm for PKP since it does not explicitly consider the penalties of the items in the iterations. At the same, it may ``stumble'' upon good quality solutions of PKP with just a negligible increase of the computational effort required for solving a KP instance. 

\medskip
Then, we compute $z(PKP_j^{+LP})$ for $j=\overline{f}+1,\ldots,n$. If the maximum of these values is not superior to $\overline{z}$, we have already certified an optimal solution for PKP (lines 4--6 in Algorithm \ref{algo:EAPKP}). 
Otherwise we greedily consider the index $k$ yielding the maximum $z(PKP_j^{+LP})$ %that is 
%\begin{equation}
%\label{eq:Compk}
%k=\underset{j}{\operatorname{arg\,max}}\;z(PKP_j^{+LP})\; \quad j= \overline{f}+1,\dots,n,
%\end{equation}
and solve $KP_1$ without items $j=1,\dots,k-1$. %This problem is denoted as $KP_2$. 
We update the values of $\overline{z}$ and $\overline{j}$ if an improving solution is found (lines 7--10 in Algorithm \ref{algo:EAPKP}).

\medskip
Finally, we compare the values $z(PKP_j^{+LP})$ with the incumbent solution value $\overline{z}$ and narrow the range of possible penalty values that may lead to an optimal solution of PKP. More precisely, we define indices $l$ and $r$
\begin{eqnarray}
l &:=& \underset{}{\operatorname{\,min}}\; \{j :  z(PKP_j^{+LP}) > \overline{z} \};
\label{eq:Compleft}\\
r &:=& \underset{}{\operatorname{\,max}}\; \{j : z(PKP_j^{+LP}) > \overline{z} \}.\label{eq:Compright}
\end{eqnarray}
The relevant interval of penalties is thus $[\Pi_{min}, \Pi_{max}]$, with $\Pi_{min}=\pi_r$ and $\Pi_{max}=\pi_l$ (line 11--14 in Algorithm \ref{algo:EAPKP}). 
If this interval is empty, the current PKP solution is also optimal and the algorithm terminates. Otherwise we get a reduced PKP with only items $j=l,\dots,n$ and the additional constraint on the penalty value $\Pi \geq \Pi_{min}$. Denote this problem by $PKP'$ and its number of items by $n'$, i.e. $n'=n-l+1$ (lines 15--17 in Algorithm \ref{algo:EAPKP}).

\medskip

This first step is expected to be fast since it relies on solving two standard knapsack problems at most and on effectively computing upper bounds for sub--problems $PKP_j^{+}$. We remark that this step is also sufficient to compute an optimal solution for a large number of instances considered in the literature.

\subsection{Step 2: A core--based dynamic programming algorithm}
\label{PKP:DP2}
%To solve the reduced problem $PKP'$, 
In this step we propose a core--based dynamic programming algorithm, hereafter denoted as $DP_2$, that constitutes a revisiting of \textit{Minknap} algorithm.
Notice that, if the size of the reduced problem $PKP'$ is reasonably small and the number of relevant penalties is large, we could otherwise solve $PKP'$ by $DP_1$ and take the best solution between $z(PKP')$ and $\overline{z}$. The choice between the algorithms is made by comparing the quantities $n'c$ and ($r - l + 1$) with the threshold parameters $T_1$ and $T_2$ (lines 18--23 in Algorithm \ref{algo:EAPKP}).

\medskip

$DP_2$ algorithm searches in $PKP'$ for solutions better than $\overline{z}$. 
%In the following we describe the algorithm after some preliminary definitions. 
Given the sorting of the items $j=1,\dots,n'$ by decreasing $\frac{p_j}{w_j}$, we define an expanding core as the interval of items $C_{a,b}=\{a,\dots,b\}$ with items $a$ and $b$ as variable extremes. Correspondingly, we define the set of 0--1 partial vectors enumerated within the core as
\begin{equation}
X_{a,b} = \{ x_j \in\{0,1\} , j \in C_{a,b} \}.
\end{equation}
%The set $X_{a,b}$ is explored by a dynamic programming with states. 
Since in any iteration of the algorithm we will have the following situation

\begin{equation}
\label{eq:PKPCore}
\overbrace{x_1,\dots,x_{a-1}}^{\text{\normalsize $x_j=1$}}, C_{a,b} ,\overbrace{x_{b+1},\dots,x_{n'}}^{\text{\normalsize $x_j=0$}}
\end{equation}
we associate each partial vector $\tilde{x} \in X_{a,b}$ with a state $(\tilde{\nu}, \tilde{\mu},\tilde{\pi}_{core},\tilde{\pi}_{tot})$ where:
\begin{enumerate}
\item $\tilde{\nu} = \sum\limits_{j=1}^{a-1} p_j + \sum\limits_{j=a}^{b}  p_j\tilde{x}_j$;
\item $\tilde{\mu} = \sum\limits_{j=1}^{a-1} w_j + \sum\limits_{j=a}^{b}  w_j\tilde{x}_j$;
\item $\tilde{\pi}_{core} = \underset{j=a,\dots,b}{\operatorname{\,max}}\; \{\pi_j : \tilde{x}_j = 1\}$;
\item $\tilde{\pi}_{tot} = \max \{\tilde{\pi}_{core}, \underset{j=1,\dots,a-1}{\operatorname{\,max}}\;\pi_j\}$.
\end{enumerate}
$\tilde{\nu}$ and  $\tilde{\mu}$ are the profits and weights of a solution with variables in the core and all variables to the left of the core;
$\tilde{\pi}_{core}$ represents the maximum penalty of the items selected in the core while $\tilde{\pi}_{tot}$ is the overall maximum penalty of the state. Each state with $\tilde{\mu} \leq c$ and $\tilde{\pi}_{tot} \geq \Pi_{min}$ represents a feasible solution  of $PKP'$ with value $\tilde{\nu} - \tilde{\pi}_{tot}$. 
We can now sketch the main steps of $DP_2$ in the following pseudo code.
The algorithm takes as input $PKP'$, the current solution ($\overline{z}, \overline{j}$) %, the penalty value $\Pi_{min}$ 
and parameters $T_3$, $\alpha$. 

\begin{algorithm}
\small
	\caption{$DP_2(PKP', \overline{z},\overline{j},T_3,\alpha)$}
	\label{algo:DP2}
	\begin{algorithmic}[1]
	%\State{\textbf{Input:} PKP' instance, $\Pi_{min}$, feasible solution ($\overline{z}$,$\overline{j}$), parameters $T_3$, $\alpha$.}

 \State{Sort items in $PKP'$ by decreasing $\frac{p_j}{w_j}$;}
 \State{$KP'$ = $PKP'$ without penalties;}
 \State{Find the split item $s'$ of $KP'$;}
\medskip 
 \State{$a=b=s'$, $C_{a,b}=\{s'\}$; $X_{a,b} = \{ (0),(1)\}$}; 
\medskip
  % \State{Reduce the set $X_{a,b}$ by upper bounds $UB^i$;}
\State{Reduce set $X_{a,b}$;}% by upper bounds (\ref{UBstates});}
\While{$X_{a,b} \neq \emptyset$ and ($b - a + 1 < n'$)}
    \State{$a \leftarrow a-1$;}
    \If{$u_0^a > \overline{z}$} 
    \If{$\tilde{u}^a > \overline{z}$}
    \State{$X_{a,b}\leftarrow Merge(a,X_{a+1,b},\Pi_{min},T_3,\alpha$);} 
    \State{Update ($\overline{z},\overline{j}$);}
		\State{Reduce set $X_{a,b}$;}% by upper bounds (\ref{UBstates});}
  \EndIf		
	
	\EndIf 	

\medskip
 \State{$b \leftarrow b+1$;},
  \If{$u_1^b > \overline{z}$} 
    \If{$\tilde{u}^b > \overline{z}$}
    \State{$X_{a,b}\leftarrow Merge(b,X_{a,b-1},\Pi_{min},T_3,\alpha$);} 
    \State{Update $(\overline{z},\overline{j}$);}
    \State{Reduce set $X_{a,b}$;}%by upper bounds (\ref{UBstates});}
  \EndIf	
	
	\EndIf 	
  
\medskip

\EndWhile 

\medskip
											
\State \Return $(\overline{z},\overline{j})$;
\end{algorithmic}
\end{algorithm}

We first sort the items of $PKP'$ by decreasing $\frac{p_j}{w_j}$ and find the split item $s'$ of the standard knapsack problem ($KP'$) induced by disregarding the penalties in $PKP'$. 
We then initialize the core with item $s'$ only 
%, i.e.\ $C_{s',s'}=\{s'\}$ 
(lines 1--4 of the pseudo code). 
Then, we enlarge the core as in \textit{Minknap} (while--loop in lines 6--23) by alternately evaluating the removal of an item $a$ from the left (lines 7--14) and the insertion of an item $b$ from the right (lines 15--22). The expansion of the core is performed by a dynamic programming with states through a procedure, denoted as \textit{Merge}, which iteratively yields undominated states in the enlarged set $X_{a,b}=X_{a+1,b} + a$ or  $X_{a,b}=X_{a,b-1} + b$. We may update the current solution ($\overline{z},\overline{j}$) if an improved solution is found while enumerating the core (lines 11 and 19). 

The dynamic programming with states is combined with an upper bound test to reduce the number of states (lines 5, 12 and 20) and two upper bound tests to limit the insertion of the variables in the core (lines 8--9 and 16--17).
The algorithm terminates whenever either the number of states is 0 or all variables have been enumerated in the core.
The ingredients of the algorithm are detailed in the following.

\subsubsection{Dynamic programming with states}
\label{subsubsec:Dominance}

The \textit{Merge} procedure performs the enumeration of the variables in the core by resembling the procedure introduced in \textit{Minknap} \cite{Pis97}, which in turn corresponds to the recursions of the \textit{primal--dual} dynamic programming algorithm in \cite{Pisinger99}. 
The proposed procedure merges, in any iteration, the current set of states $X$ and $X + d$, where $X + d$ is set $X$ with profits, weights and penalties of the states updated according to the removal/insertion of item $d$ from/in the knapsack. 
In the merging operation the states are kept ordered by increasing weights so as to effectively apply a dominance rule for PKP. 

The classical dominance rule in KP considers the weights and profits of the states. 
For PKP, let us define the quantity $\rho = \nu - \max \{\pi_{core}, \Pi_{min}\}$ which represents the difference between the profit of a state and the minimum penalty that the state must have for yielding an optimal solution. This penalty corresponds to the maximum between $\Pi_{min}$ and $\pi_{core}$ since, due to the enumeration of the core, for any state $\pi_{core}$ constitutes a minimum penalty value in all states originating from it while $\Pi_{min}$ is the minimum penalty required in any solution with a value greater than $\overline{z}$. We introduce the following dominance rule for two generic states $i$ and $j$.
\begin{prop} 
\label{PKPdominance}
Given states $i$ and $j$ and their quantities fulfilling
\begin{equation}
\label{EqPKPdom}
\mu^i \leq \mu^j, \qquad \nu^i \geq \nu^j, \qquad \rho^{i} \geq \rho^{j}.
\end{equation}
Then state $j$ is said to be dominated by state $i$ and can be discarded in the search for an optimal solution of PKP.
\end{prop}

\begin{proof}
The first two conditions represent the dominance of state $i$ in the standard KP.
The condition $\rho^{i} \geq \rho^{j}$
implies that all successive states deriving from state $i$ and possibly optimal for $PKP$ (i.e.\ with a penalty greater 
than $\Pi_{min}$) would have a no worse solution value than those deriving from state $j$.
\end{proof}

We remark that, given the presence of penalties, the ordering of states by increasing weights does not imply the ordering of the profits as in \textit{Minknap}. To better detect situations of dominance, we apply the rule in Proposition \ref{PKPdominance} by comparing each state with a number of states (with a lower weight) given by parameter $\alpha$. 

\medskip

Whenever only the condition involving the penalties prevents the fathoming of state $j$, we may combine the dominance rule with an upper bound on state $j$ depending on a penalty value $\pi > \max \{\pi_{core}^j, \Pi_{min}\}$.\\
This upper bound, denoted by $UB(\pi)^j$, is computed as follows. We first solve the linear relaxation of the KP induced by packing the items selected in the core for state $j$ and by disregarding the items outside the core with a higher penalty than $\pi$. From the optimal solution value  of this problem we then subtract the maximum value between $\pi_{core}^j$ and $\Pi_{min}$. The following proposition holds

\begin{prop} 
\label{PKPdominance2}
Given two states $i$ and $j$ and the quantities
\begin{equation}
\label{PKPdomin}
\mu^i \leq \mu^j, \qquad \nu^i \geq \nu^j, \qquad \rho^{i} < \rho^{j}, 
\end{equation}
consider the maximum penalty $\hat{\pi}$ which would not induce the dominance of state $i$ according to (\ref{EqPKPdom}), i.e.\ $\hat{\pi} = \underset{}{\operatorname{\,max}}\; \{\pi : \nu^j - \pi > \rho^i \}$. If $UB(\hat{\pi})^j \leq \overline{z}$, then state $j$ can be discarded.
\end{prop}

\begin{proof}
We analyze the solution values deriving from state $j$ when the overall maximum penalty is upper bounded by a quantity $\pi'$.  
For any $\pi' \leq \hat{\pi}$, since $UB(\hat{\pi})^j \leq \overline{z}$ we can discard state $j$ because all states deriving from state $j$ cannot reach a solution values greater than $\overline{z}$.
Likewise, we can as well discard state $j$ if $\pi' > \hat{\pi}$ since this condition would induce a dominance of state $i$.
\end{proof}
Computing $UB(\pi)$ has complexity $O(n)$ and would be time--consuming if the number of states involved is sufficiently large. Thus, we calculate this bound only if the number of states exceeds the threshold value $T_3$. 

\subsubsection{Reduction of the states}
\label{subsubsec:Red}

To further reduce the set of states, we also perform an upper bound test in constant time for each state. %
In any iteration, we compute the following upper bound for a state $i$ associated with $X_{a,b}$:
\begin{align}
&UB^i= \left\{ \begin{array}{ll} 
                             \rho^i + (c -\mu^i)\frac{p_{b+1}}{w_{b+1}} & \mbox{ if } \mu^i \leq c \\[2ex]
                             \rho^i + (c -\mu^i)\frac{p_{a-1}}{w_{a-1}} &\mbox{ if }  \mu^i > c       
                                 \end{array} \right.
\label{UBstates}
\end{align}
and discard state $i$ if $UB^i \leq \overline{z}$. These upper bounds are computed by replacing the integrality constraint on $x_{a-1}$ and $x_{b+1}$ with $x_{a-1} \geq 0$ and $x_{b+1} \geq 0$ and by disregarding the penalty values of the variables outside the core.  

\subsubsection{Upper bound tests on the variables outside the core}
\label{subsubsec:Red}
Since the insertion of variables in the core may be computationally expensive, we perform two upper bound tests whenever an item $j$ is candidate to be included in the core. \\
We first compute similar bounds to the ones proposed in \cite{DeHa80} for KP. Let us denote by $u_0^j$ an upper bound on $PKP'$ without item $j$. 
Also, let us denote by $u_1^j$ the upper bound when item $j$ is packed. 
The following bounds are computed in constant time for each item $j$:
\begin{eqnarray}
u_0^j  &=& p' - p_j - \Pi_{min} + (c - w' + w_j)\frac{p_{s'}}{w_{s'}} \quad j= 1,\dots,s'-1
\\
u_1^j &=& p' + p_j  - \max \{\pi_j, \Pi_{min}\} + (c - w' - w_j)\frac{p_{s'}}{w_{s'}} 
\ j= s'+1,\dots,n
\end{eqnarray}
Here $w'=\sum\limits_{j=1}^{s'-1} w_j $ and $p'=\sum\limits_{j=1}^{s'-1} p_j $  represent the weight and the profit of the split solution of $KP'$. If $u_0^j$ (resp.\ $u_1^j$) $\leq \overline{z}$, we can fix variable $x_j = 1$ (resp.\ $x_j = 0$).  

\medskip

In cascade, we may perform a second test by computing a stronger upper bound in linear time with the number of states.
As in \textit{Minknap}, we evaluate the impacts of removing (inserting) item $j$ with $j < s'$ ($j > s'$) in all states in the current set $X$, namely we derive states $i \in X + j$ and compute upper bounds (\ref{UBstates}) on these states. 
A valid upper bound for item $j$, denoted as $\tilde{u}^j$, is constituted by the maximum of these bounds. As pointed out in \cite{Pis97}, $\tilde{u}^j$ can be seen as a generalization of the enumerative bound in \cite{MarToth88}.
If $\tilde{u}^j \leq \overline{z}$, then variable $x_j$ is fixed to the value taken in the split solution. 

\medskip
After this second step, the optimal solution value $z^*$ and the optimal leading item  $j^*$ are returned. The optimal solution set of items can be determined by solving the standard knapsack problem $PKP_{j^*}^+$.

\section{Computational results}
\label{sec:CompResPKP}

All tests were performed on an Intel i7 CPU @ 2.4 GHz with 8 GB of RAM. The code was implemented in the C++ programming language.
We generated the instances according to the generation scheme proposed in \cite{CeRi06}. We considered two types of weights: $a1$ and $a2$.  In the former, the weights are randomly distributed in $[1,R]$,  with $R$ being an arbitrary parameter. In the latter, the weights are equal to $\frac{R}{2}$ + $v$, with $v$ uniformly distributed in $[0,\frac{R}{2} ]$. Basically, small weights are not considered in $a2$. 

We generated 8 classes of penalties ($\pi1, \dots, \pi8$) and 7 classes of profits ($p1, \dots, p7$) according to different correlations of penalties/profits with the weights, as illustrated in Table~\ref{tab:CorrPKP}. 
The first 6 correlations correspond to classical correlations in KP instances. In class $\pi7$ penalties $\pi_j$ are equal to $R - w_j + 1$ (constant perimeter correlation) while in class $\pi8$ we set $\pi_j = \frac{R}{w_j}$ (constant area correlation). In class $p7$ we set $p_j = \pi_jw_j$. Finally, three different values of the ratio $\tau$ between the knapsack capacity and the sum of the weights of the items are considered: 0.5, 0.1 and 0.01.

\begin{table}[h]
	\centering
	\small
	 %\scriptsize
\begin{tabular}{rll}
  \hline
$\pi$ type & Correlation & $p$ type\\ \hline
$\pi1$ & No correlation & $p1$\\
$\pi2$ & Weak correlation & $p2$\\
$\pi3$ & Strong correlation & $p3$\\
$\pi4$ & Inverse strong correlation & $p4$\\
$\pi5$ & Almost strong correlation & $p5$\\
$\pi6$ & Subset-sum correlation & $p6$\\
$\pi7$ & Constant perimeter &	 \\
$\pi8$ & Constant area &	 \\
		& Profit = area & $p7$\\

 \hline
\end{tabular}
		\caption{Correlation types from \cite{CeRi06}.} 
		\label{tab:CorrPKP}
\end{table}

We first generated instances with 1000 items and $R = 1000$. Within each category, 
five instances were tested for a total of 1680 instances. Similarly, we generated 1680 instances with 10000 items and $R = 10000$. We compared the solutions reached by the proposed exact approach, the algorithm in \cite{CeRi06} and CPLEX 12.5 running on model (PKP). 
After some preliminary test runs, we chose the following parameter values for our approach: $\alpha = 15$, $T_1 = 5*10^9$, $T_2 = \frac{n}{10}$, $T_3 = 3*10^6$. The parameters of the ILP solver were set to their default values. 

The results are summarized in Tables \ref{tab:CPUTimePKP2} and \ref{tab:CPUTimePKP3} in terms of average, maximum CPU time and number of optima obtained within a time limit of 100 seconds. The average CPU times consider also the cases where the time limit is reached. 
The results are aggregated by profit classes and weight types. Each entry in the tables reports the results over 120 instances.
\ifdefined\LONG
Detailed results for all correlations and capacity ratios are given at the end of this section.
\else{}
%In this paper we restrict ourselves to reporting summary results.
Detailed results for all correlations and capacity ratios are available in \cite{longversion}.
\fi{}

\medskip
%\begin{table}[H]
\begin{table}[ht]
	\centering
	\scriptsize
	\scalebox{0.8}{
\begin{tabular}{|cc|*{3}{c|}*{3}{c|}*{3}{c|}}
  \hline %\hline
    $n = 1000$  &   & \multicolumn{3}{|c|}{CPLEX 12.5} & \multicolumn{3}{|c|}{Algorithm in \cite{CeRi06}} & \multicolumn{3}{|c|}{Proposed exact approach}  \\ \hline
  Profit & Weight &   Average &  Max  &&   Average &  Max  &&   Average &  Max  &\\
    class  &  type &   time (s)& time (s) & \#Opt&   time (s)& time (s) & \#Opt&   time (s)& time (s) & \#Opt\\										
									\hline		 %\hline
$p1$ & $a1$ & 0.18  & 0.25   & 120 & 0.00 & 0.01 & 120 & 0.00 & 0.00 & 120 \\
         & $a2$ & 0.19  & 0.48   & 120 & 0.00 & 0.01 & 120 & 0.00 & 0.03 & 120 \\ \hline
$p2$ & $a1$ & 0.39  & 1.69   & 120 & 0.00 & 0.01 & 120 & 0.00 & 0.10 & 120 \\
         & $a2$ & 1.12  & 6.32   & 120 & 0.00 & 0.01 & 120 & 0.01 & 0.27 & 120 \\ \hline
$p3$ & $a1$ & 3.90  & 100.00 & 117 & 0.04 & 0.89 & 120 & 0.01 & 0.40 & 120 \\
         & $a2$ & 6.83  & 100.00 & 117 & 0.50 & 8.02 & 120 & 0.02 & 0.29 & 120 \\ \hline
$p4$ & $a1$ & 59.56 & 100.00 & 57  & 0.10 & 1.38 & 120 & 0.02 & 0.18 & 120 \\
         & $a2$ & 66.97 & 100.00 & 46  & 0.28 & 7.75 & 120 & 0.04 & 0.26 & 120 \\ \hline
$p5$ & $a1$ & 4.13  & 100.00 & 117 & 0.03 & 0.89 & 120 & 0.02 & 0.20 & 120 \\
         & $a2$ & 14.97 & 100.00 & 114 & 0.41 & 4.50 & 120 & 0.07 & 1.40 & 120 \\ \hline
$p6$ & $a1$ & 2.67  & 90.38  & 120 & 0.00 & 0.06 & 120 & 0.00 & 0.03 & 120 \\
         & $a2$ & 2.97  & 13.57  & 120 & 0.01 & 0.15 & 120 & 0.01 & 0.08 & 120 \\ \hline
$p7$ & $a1$ & 27.58 & 100.00 & 89  & 0.00 & 0.01 & 120 & 0.00 & 0.02 & 120 \\
         & $a2$ & 38.24 & 100.00 & 75  & 0.01 & 0.06 & 120 & 0.01 & 0.07 & 120\\ \hline						
\end{tabular}}
		\caption{Summary results for instances with 1000 items and different correlations between profits and weights: time (s) and number of optima over 120 instances.} 
		\label{tab:CPUTimePKP2}
\end{table}

%\begin{table}[H]
\begin{table}[ht]
	\centering
	\scriptsize
	\scalebox{0.8}{
\begin{tabular}{|cc|*{3}{c|}*{3}{c|}*{3}{c|}}
  \hline %\hline
    $n = 10000$  &   & \multicolumn{3}{|c|}{CPLEX 12.5} & \multicolumn{3}{|c|}{Algorithm in \cite{CeRi06}} & \multicolumn{3}{|c|}{Proposed exact approach}  \\ \hline
  Profit & Weight &   Average &  Max  &&   Average &  Max  &&   Average &  Max  &\\
    class  &  type &   time (s)& time (s) & \#Opt&   time (s)& time (s) & \#Opt&   time (s)& time (s) & \#Opt\\										
									\hline		 %\hline
$p1$ & $a1$ & 0.95  & 2.68   & 120 & 0.01  & 0.02   & 120 & 0.01 & 0.03  & 120 \\
         & $a2$ & 4.19  & 100.00 & 117 & 0.01  & 0.04   & 120 & 0.01 & 0.03  & 120 \\ \hline
$p2$ & $a1$ & 5.41  & 33.84  & 120 & 0.01  & 0.02   & 120 & 0.02 & 0.13  & 120 \\
         & $a2$ & 12.43 & 100.00 & 114 & 0.01  & 0.07   & 120 & 0.04 & 0.73  & 120 \\ \hline
$p3$ & $a1$ & 44.61 & 100.00 & 74  & 25.59 & 100.00 & 96  & 2.59 & 58.46 & 120 \\
         & $a2$ & 74.13 & 100.00 & 46  & 48.26 & 100.00 & 74  & 5.53 & 29.00 & 120 \\ \hline
$p4$ & $a1$& 91.60 & 100.00 & 11  & 17.68 & 100.00 & 106 & 2.81 & 18.88 & 120 \\
         & $a2$ & 94.69 & 100.00 & 7   & 26.34 & 100.00 & 106 & 7.82 & 80.02 & 120 \\ \hline
$p5$ & $a1$ & 25.45 & 100.00 & 95  & 10.70 & 100.00 & 113 & 2.66 & 70.71 & 120 \\
         & $a2$ & 65.99 & 100.00 & 48  & 23.74 & 100.00 & 101 & 7.04 & 58.60 & 120 \\ \hline
$p6$ & $a1$ & 83.08 & 100.00 & 58  & 0.67  & 40.17  & 120 & 0.22 & 5.38  & 120 \\
         & $a2$ & 81.05 & 100.00 & 44  & 3.62  & 100.00 & 119 & 1.95 & 16.65 & 120 \\ \hline
$p7$ & $a1$ & 51.00 & 100.00 & 63  & 0.17  & 0.94   & 120 & 0.42 & 2.50  & 120 \\
         & $a2$ & 40.12 & 100.00 & 75  & 0.54  & 3.94   & 120 & 1.41 & 11.46 & 120\\ \hline						
\end{tabular}}
		\caption{Summary results for instances with 10000 items and different correlations between profits and weights: time (s) and number of optima over 120 instances.} 
		\label{tab:CPUTimePKP3}
\end{table}

From Tables \ref{tab:CPUTimePKP2} and \ref{tab:CPUTimePKP3} we see that, for the instances with 1000 items, both the proposed exact approach and the algorithm in \cite{CeRi06} outperform CPLEX 12.5 which does not reach all the optima within the time limit. Although the performances of the algorithms are similar, we note that our approach generally performs slightly better and requires 1.4 seconds at most for solving to optimality all instances.

In the largest instances with 10000 items, our algorithm strongly outperforms both CPLEX 12.5 and the algorithm in \cite{CeRi06}. Our approach is capable of reaching all optima with limited CPU time (80 seconds at most for an instance in class $p4$) while the solver and the competing algorithm run out of time for several large instances. The largest differences in computational times involve instances in classes $p3$, $p4$ and $p5$. 

The most challenging instances for our algorithm turned out to be the ones without small weights ($a2$). In general, the absence of small weights might increase the computational effort required for solving even standard KP instances (as pointed out, e.g., in \cite{CeRi06}) and this is presumably the reason of the increase in CPU times of our algorithm as well. 

In many instances, the first main step relying on solving standard KPs is sufficient to certificate an optimal solution for PKP. Indeed, this constitutes a remarkable strength of our procedures. In Tables \ref{tab:CPUTimePKP4} and \ref{tab:CPUTimePKP5} we report the percentage of the optimal solutions already computed by the first step of the procedure for the instances with 1000 and 10000 items respectively.  Averaged computational times (\% of the total CPU time) of the two steps of our approach are also reported. Finally, we report the average and maximum values (in thousands) of the maximum number of states reached by $DP_2$ algorithm in each instance. We point out that $DP_1$ algorithm is called a limited number of times with respect to $DP_2$ (5\% of the cases) and mainly in the smallest instances with 1000 items.  

%\begin{table}[H]
\begin{table}[ht]
	\centering
	\scriptsize
\begin{tabular}{|cc|c|*{2}{c|}*{2}{c|}}
  \hline %\hline
    Proposed exact approach &  & \textit{Step 1}&\multicolumn{2}{|c|}{\textit{Step 1} and} & \multicolumn{2}{|c|}{Max number of states} \\ 
    	($n = 1000$)	  &  & only & \multicolumn{2}{|c|}{ \textit{Step 2}} & \multicolumn{2}{|c|}{in $DP_2$} \\ \hline
    		%	  &  &  & \textit{Step 1}& \textit{Step 2} & \multicolumn{2}{|c|}{} \\ \hline
  Profit & Weight &  \#Opt & Time & Time & Average &  Max  \\
    class  &  type &(\%)& (\%)  & (\%)&  (x$10^3$) & (x$10^3$) \\										
									\hline		 %\hline
$p1$ & $a1$ & 72.5 & 54.0 & 46.0 & 0.1  & 0.2  \\
          & $a2$ & 60.0 & 56.8 & 43.2 & 0.1  & 0.5  \\ \hline
$p2$ & $a1$ & 43.3 & 50.6 & 49.4 & 0.8  & 16.0 \\
         & $a2$ & 49.2 & 49.4 & 50.6 & 1.6  & 20.5 \\ \hline
$p3$ & $a1$ & 69.2 & 58.1 & 41.9 & 2.8  & 56.2 \\
         & $a2$ & 52.5 & 78.0 & 22.0 & 1.8  & 6.8  \\ \hline
$p4$ & $a1$ &48.3 & 78.0 & 22.0 & 2.1  & 19.6 \\
         & $a2$ &57.5 & 73.7 & 26.3 & 2.6  & 14.9 \\ \hline
$p5$ & $a1$ &22.5 & 43.0 & 57.0 & 5.1  & 36.1 \\
         & $a2$ &24.2 & 43.5 & 56.5 & 13.4 & 58.9 \\ \hline
$p6$ & $a1$ & 82.5 & 41.5 & 58.5 & 4.9  & 25.2 \\
         & $a2$ &27.5 & 36.0 & 64.0 & 7.6  & 47.8 \\ \hline
$p7$ & $a1$ & 59.2 & 51.0 & 49.0 & 0.9  & 3.2  \\
         & $a2$ & 58.3 & 48.5 & 51.5 & 2.3  & 7.5	\\ \hline			
\end{tabular}
		\caption{Numerical insights of the proposed exact approach for instances with 1000 items.} 
		\label{tab:CPUTimePKP4}
\end{table}

%\begin{table}[H]
\begin{table}[ht]
	\centering
	\scriptsize
\begin{tabular}{|cc|c|*{2}{c|}*{2}{c|}}
  \hline %\hline
   Proposed exact approach &  & \textit{Step 1}&\multicolumn{2}{|c|}{\textit{Step 1} and} & \multicolumn{2}{|c|}{Max number of states} \\ 
    	($n = 10000$)	  &  & only & \multicolumn{2}{|c|}{ \textit{Step 2}} & \multicolumn{2}{|c|}{in $DP_2$} \\ \hline
    		%	  &  &  & \textit{Step 1}& \textit{Step 2} & \multicolumn{2}{|c|}{} \\ \hline
  Profit & Weight &  \#Opt & Time & Time & Average &  Max  \\
    class  &  type &(\%)& (\%)  & (\%)&  (x$10^3$) & (x$10^3$) \\										
									\hline		 %\hline
$p1$ & $a1$ & 85.0 & 79.1 & 20.9 & 0.7   & 6.3    \\
         & $a2$ & 69.2 & 62.3 & 37.7 & 1.5   & 5.5    \\ \hline
$p2$ & $a1$ & 50.0 & 68.3 & 31.7 & 3.4   & 66.2   \\
         & $a2$ & 52.5 & 48.6 & 51.4 & 19.5  & 166.1  \\ \hline
$p3$ & $a1$ &77.5 & 58.0 & 42.0 & 146.3 & 1115.8 \\
         & $a2$ & 56.7 & 90.1 & 9.9  & 45.8  & 247.4  \\ \hline
$p4$ & $a1$ &73.3 & 74.5 & 25.5 & 119.9 & 713.0  \\
         & $a2$ & 79.2 & 83.8 & 16.2 & 139.8 & 885.8  \\ \hline
$p5$ & $a1$ &27.5 & 39.3 & 60.7 & 164.5 & 1244.3 \\
         & $a2$ & 25.0 & 39.2 & 60.8 & 444.4 & 3088.3 \\ \hline
$p6$ & $a1$ &83.3 & 32.0 & 68.0 & 187.5 & 700.5  \\
         & $a2$ & 33.3 & 26.8 & 73.2 & 321.6 & 1292.2 \\ \hline
$p7$ & $a1$ &55.0 & 60.3 & 39.7 & 81.1  & 493.3  \\
         & $a2$ & 63.3 & 60.2 & 39.8 & 145.6 & 764.9\\ \hline		
\end{tabular}
		\caption{Numerical insights of the proposed exact approach for instances with 10000 items.} 
		\label{tab:CPUTimePKP5}
\end{table}

The results in the tables illustrate the effectiveness of the first step in solving PKP instances.  
Usually more than 50\% of the instances %(with the exception of class $p5$) 
are solved to optimality within this step. When both steps are involved, the computational effort is on average equally distributed. We note however an increase of the percentages of the second step in classes $p5$ and $p6$. The number of states is in general reasonably limited allowing our algorithm to effectively solve all instances considered. The largest values of the number of states (with a maximum of about 3 millions) are reached in the instances with 10000 items.

\ifdefined\LONG
In the following we also list detailed results for all correlations and capacity ratios.

\begin{table}[H]
	\centering
	\scriptsize
	\scalebox{0.85}{
\begin{tabular}{|ccc|*{3}{c|}*{3}{c|}*{3}{c|}}
  \hline 
    $n = 1000$  &  & & \multicolumn{3}{|c|}{CPLEX 12.5} & \multicolumn{3}{|c|}{Algorithm in \cite{CeRi06}} & \multicolumn{3}{|c|}{Proposed exact approach}  \\ \hline
  Weight & &Profit &   Average &  Max  &&   Average &  Max  &&   Average &  Max  &\\
  type  & $\tau$ &class &   time (s)& time (s) & \#Opt&   time (s)& time (s) & \#Opt&   time (s)& time (s) & \#Opt\\										
									\hline		
$a1$ & 0.5& $p1$ & 0.19  & 0.25   & 40 & 0.00 & 0.00 & 40 & 0.00 & 0.00 & 40 \\
   &      & $p2$ & 0.32  & 0.92   & 40 & 0.00 & 0.01 & 40 & 0.00 & 0.10 & 40 \\
   &      & $p3$ & 10.59 & 100.00 & 37 & 0.10 & 0.89 & 40 & 0.03 & 0.40 & 40 \\
   &      & $p4$ & 62.27 & 100.00 & 16 & 0.14 & 1.34 & 40 & 0.02 & 0.05 & 40 \\
   &      & $p5$ & 11.08 & 100.00 & 37 & 0.09 & 0.89 & 40 & 0.04 & 0.20 & 40 \\
   &      & $p6$ & 2.29  & 18.94  & 40 & 0.00 & 0.00 & 40 & 0.00 & 0.02 & 40 \\
   &      & $p7$ & 36.28 & 100.00 & 26 & 0.00 & 0.01 & 40 & 0.00 & 0.01 & 40 \\
   & 0.1  & $p1$ & 0.20  & 0.25   & 40 & 0.00 & 0.01 & 40 & 0.00 & 0.00 & 40 \\
   &      & $p2$ & 0.49  & 1.13   & 40 & 0.00 & 0.00 & 40 & 0.00 & 0.05 & 40 \\
   &      & $p3$ & 0.54  & 1.44   & 40 & 0.02 & 0.11 & 40 & 0.01 & 0.04 & 40 \\
   &      & $p4$ & 74.13 & 100.00 & 11 & 0.12 & 1.38 & 40 & 0.02 & 0.18 & 40 \\
   &      & $p5$ & 0.93  & 10.24  & 40 & 0.01 & 0.06 & 40 & 0.01 & 0.05 & 40 \\
   &      & $p6$ & 4.23  & 90.38  & 40 & 0.00 & 0.01 & 40 & 0.00 & 0.03 & 40 \\
   &      & $p7$ & 28.61 & 100.00 & 29 & 0.00 & 0.01 & 40 & 0.01 & 0.02 & 40 \\
   & 0.01 & $p1$ & 0.14  & 0.18   & 40 & 0.00 & 0.01 & 40 & 0.00 & 0.00 & 40 \\
   &      & $p2$ & 0.35  & 1.69   & 40 & 0.00 & 0.00 & 40 & 0.00 & 0.01 & 40 \\
   &      & $p3$ & 0.58  & 8.34   & 40 & 0.00 & 0.03 & 40 & 0.00 & 0.01 & 40 \\
   &      & $p4$ & 42.27 & 100.00 & 30 & 0.04 & 0.20 & 40 & 0.01 & 0.03 & 40 \\
   &      & $p5$ & 0.38  & 3.65   & 40 & 0.00 & 0.01 & 40 & 0.00 & 0.01 & 40 \\
   &      & $p6$ & 1.48  & 4.05   & 40 & 0.00 & 0.06 & 40 & 0.00 & 0.01 & 40 \\
   &      & $p7$ & 17.87 & 100.00 & 34 & 0.00 & 0.01 & 40 & 0.00 & 0.02 & 40 \\
$a2$ & 0.5& $p1$ & 0.18  & 0.32   & 40 & 0.00 & 0.00 & 40 & 0.00 & 0.00 & 40 \\
   &      & $p2$ & 0.47  & 3.19   & 40 & 0.00 & 0.01 & 40 & 0.00 & 0.00 & 40 \\
   &      & $p3$ & 10.66 & 100.00 & 37 & 0.89 & 8.02 & 40 & 0.03 & 0.11 & 40 \\
   &      & $p4$ & 53.29 & 100.00 & 19 & 0.50 & 7.75 & 40 & 0.04 & 0.11 & 40 \\
   &      & $p5$ & 23.32 & 100.00 & 36 & 0.65 & 4.50 & 40 & 0.12 & 1.40 & 40 \\
   &      & $p6$ & 4.65  & 13.57  & 40 & 0.00 & 0.01 & 40 & 0.02 & 0.08 & 40 \\
   &      & $p7$ & 17.76 & 100.00 & 33 & 0.00 & 0.01 & 40 & 0.01 & 0.03 & 40 \\
   & 0.1  & $p1$ & 0.22  & 0.33   & 40 & 0.00 & 0.00 & 40 & 0.00 & 0.00 & 40 \\
   &      & $p2$ & 1.52  & 6.32   & 40 & 0.00 & 0.01 & 40 & 0.01 & 0.27 & 40 \\
   &      & $p3$ & 3.40  & 41.67  & 40 & 0.41 & 2.17 & 40 & 0.03 & 0.29 & 40 \\
   &      & $p4$ & 87.51 & 100.00 & 6  & 0.31 & 7.30 & 40 & 0.04 & 0.26 & 40 \\
   &      & $p5$ & 10.75 & 100.00 & 38 & 0.38 & 1.63 & 40 & 0.07 & 0.38 & 40 \\
   &      & $p6$ & 1.91  & 5.30   & 40 & 0.00 & 0.08 & 40 & 0.01 & 0.04 & 40 \\
   &      & $p7$ & 35.62 & 100.00 & 26 & 0.01 & 0.03 & 40 & 0.01 & 0.07 & 40 \\
   & 0.01 & $p1$ & 0.18  & 0.48   & 40 & 0.00 & 0.01 & 40 & 0.00 & 0.03 & 40 \\
   &      & $p2$ & 1.38  & 4.26   & 40 & 0.00 & 0.01 & 40 & 0.00 & 0.01 & 40 \\
   &      & $p3$ & 6.43  & 18.53  & 40 & 0.20 & 1.17 & 40 & 0.02 & 0.03 & 40 \\
   &      & $p4$ & 60.12 & 100.00 & 21 & 0.04 & 0.14 & 40 & 0.03 & 0.05 & 40 \\
   &      & $p5$ & 10.84 & 60.54  & 40 & 0.19 & 0.84 & 40 & 0.02 & 0.04 & 40 \\
   &      & $p6$ & 2.36  & 6.48   & 40 & 0.03 & 0.15 & 40 & 0.01 & 0.07 & 40 \\
   &      & $p7$ & 61.34 & 100.00 & 16 & 0.01 & 0.06 & 40 & 0.02 & 0.06 & 40 \\ \hline		
\end{tabular}}
		\caption{Computational results for instances with 1000 items and different correlations between profits and weights: time (s) and number of optima over 40 instances.} 
		\label{tab:CPUTimePKP6}
\end{table}

\begin{table}[H]
	\centering
	\scriptsize
		\scalebox{0.85}{
\begin{tabular}{|ccc|*{3}{c|}*{3}{c|}*{3}{c|}}
  \hline 
    $n = 10000$  &  & & \multicolumn{3}{|c|}{CPLEX 12.5} & \multicolumn{3}{|c|}{Algorithm in \cite{CeRi06}} & \multicolumn{3}{|c|}{Proposed exact approach}  \\ \hline
  Weight & &Profit &   Average &  Max  &&   Average &  Max  &&   Average &  Max  &\\
  type  & $\tau$ &class &   time (s)& time (s) & \#Opt&   time (s)& time (s) & \#Opt&   time (s)& time (s) & \#Opt\\										
									\hline		
$a1$ & 0.5& $p1$ & 1.00   & 2.28   & 40 & 0.01  & 0.02   & 40 & 0.02  & 0.03  & 40 \\
   &      & $p2$ & 3.73   & 9.40   & 40 & 0.01  & 0.02   & 40 & 0.02  & 0.03  & 40 \\
   &      & $p3$ & 62.62  & 100.00 & 16 & 45.29 & 100.00 & 23 & 6.65  & 58.46 & 40 \\
   &      & $p4$ & 82.03  & 100.00 & 8  & 17.43 & 100.00 & 34 & 2.79  & 10.73 & 40 \\
   &      & $p5$ & 53.44  & 100.00 & 20 & 24.99 & 100.00 & 33 & 7.16  & 70.71 & 40 \\
   &      & $p6$ & 90.65  & 100.00 & 12 & 0.01  & 0.06   & 40 & 0.23  & 2.56  & 40 \\
   &      & $p7$ & 21.31  & 100.00 & 33 & 0.09  & 0.33   & 40 & 0.15  & 0.51  & 40 \\
   & 0.1  & $p1$ & 1.12   & 2.68   & 40 & 0.01  & 0.01   & 40 & 0.01  & 0.02  & 40 \\
   &      & $p2$ & 11.05  & 33.84  & 40 & 0.01  & 0.02   & 40 & 0.02  & 0.13  & 40 \\
   &      & $p3$ & 55.29  & 100.00 & 21 & 26.62 & 100.00 & 33 & 1.02  & 8.73  & 40 \\
   &      & $p4$ & 95.24  & 100.00 & 2  & 16.64 & 100.00 & 37 & 3.54  & 18.88 & 40 \\
   &      & $p5$ & 20.70  & 100.00 & 35 & 6.24  & 78.90  & 40 & 0.76  & 6.13  & 40 \\
   &      & $p6$ & 77.25  & 100.00 & 24 & 0.09  & 1.97   & 40 & 0.33  & 5.38  & 40 \\
   &      & $p7$ & 64.61  & 100.00 & 15 & 0.27  & 0.94   & 40 & 0.57  & 2.26  & 40 \\
   & 0.01 & $p1$ & 0.73   & 1.28   & 40 & 0.00  & 0.01   & 40 & 0.00  & 0.01  & 40 \\
   &      & $p2$ & 1.45   & 2.74   & 40 & 0.00  & 0.02   & 40 & 0.01  & 0.02  & 40 \\
   &      & $p3$ & 15.92  & 100.00 & 37 & 4.86  & 59.95  & 40 & 0.12  & 1.00  & 40 \\
   &      & $p4$ & 97.53  & 100.00 & 1  & 18.97 & 100.00 & 35 & 2.10  & 14.20 & 40 \\
   &      & $p5$ & 2.22   & 6.56   & 40 & 0.86  & 7.79   & 40 & 0.07  & 0.56  & 40 \\
   &      & $p6$ & 81.35  & 100.00 & 22 & 1.91  & 40.17  & 40 & 0.11  & 1.69  & 40 \\
   &      & $p7$ & 67.09  & 100.00 & 15 & 0.14  & 0.56   & 40 & 0.54  & 2.50  & 40 \\
$a2$ & 0.5& $p1$ & 1.93   & 11.22  & 40 & 0.01  & 0.02   & 40 & 0.02  & 0.02  & 40 \\
   &      & $p2$ & 3.86   & 11.72  & 40 & 0.01  & 0.02   & 40 & 0.02  & 0.03  & 40 \\
   &      & $p3$ & 70.77  & 100.00 & 18 & 53.23 & 100.00 & 23 & 8.09  & 25.55 & 40 \\
   &      & $p4$ & 86.05  & 100.00 & 6  & 24.86 & 100.00 & 35 & 8.72  & 26.06 & 40 \\
   &      & $p5$ & 60.26  & 100.00 & 19 & 24.30 & 100.00 & 32 & 9.93  & 58.60 & 40 \\
   &      & $p6$ & 86.29  & 100.00 & 7  & 0.08  & 0.95   & 40 & 2.62  & 16.65 & 40 \\
   &      & $p7$ & 14.30  & 100.00 & 36 & 0.31  & 2.23   & 40 & 0.48  & 3.34  & 40 \\
   & 0.1  & $p1$ & 1.27   & 2.19   & 40 & 0.01  & 0.02   & 40 & 0.01  & 0.02  & 40 \\
   &      & $p2$ & 10.99  & 29.04  & 40 & 0.01  & 0.02   & 40 & 0.04  & 0.66  & 40 \\
   &      & $p3$ & 66.31  & 100.00 & 19 & 50.95 & 100.00 & 24 & 4.33  & 14.29 & 40 \\
   &      & $p4$ & 98.01  & 100.00 & 1  & 34.25 & 100.00 & 35 & 11.28 & 80.02 & 40 \\
   &      & $p5$ & 54.25  & 100.00 & 21 & 27.97 & 100.00 & 35 & 5.97  & 40.59 & 40 \\
   &      & $p6$ & 82.95  & 100.00 & 13 & 0.29  & 1.47   & 40 & 2.28  & 7.87  & 40 \\
   &      & $p7$ & 51.33  & 100.00 & 20 & 0.77  & 3.94   & 40 & 1.83  & 11.46 & 40 \\
   & 0.01 & $p1$ & 9.37   & 100.00 & 37 & 0.01  & 0.04   & 40 & 0.01  & 0.03  & 40 \\
   &      & $p2$ & 22.45  & 100.00 & 34 & 0.01  & 0.07   & 40 & 0.06  & 0.73  & 40 \\
   &      & $p3$ & 85.32  & 100.00 & 9  & 40.59 & 100.00 & 27 & 4.18  & 29.00 & 40 \\
   &      & $p4$ & 100.00 & 100.00 & 0  & 19.92 & 100.00 & 36 & 3.46  & 15.36 & 40 \\
   &      & $p5$ & 83.47  & 100.00 & 8  & 18.96 & 100.00 & 34 & 5.22  & 20.51 & 40 \\
   &      & $p6$ & 73.90  & 100.00 & 24 & 10.50 & 100.00 & 39 & 0.96  & 3.90  & 40 \\
   &      & $p7$ & 54.74  & 100.00 & 19 & 0.55  & 3.47   & 40 & 1.92  & 8.74  & 40 \\ \hline	
\end{tabular}}
		\caption{Computational results for instances with 10000 items and different correlations between profits and weights: time (s) and number of optima over 40 instances.} 
		\label{tab:CPUTimePKP7}
\end{table}

\begin{table}[H]
	\centering
	\scriptsize
	\scalebox{0.85}{
\begin{tabular}{|ccc|*{3}{c|}*{3}{c|}*{3}{c|}}
  \hline 
    $n = 1000$  &  & & \multicolumn{3}{|c|}{CPLEX 12.5} & \multicolumn{3}{|c|}{Algorithm in \cite{CeRi06}} & \multicolumn{3}{|c|}{Proposed exact approach}  \\ \hline
  Weight & &Penalty &   Average &  Max  &&   Average &  Max  &&   Average &  Max  &\\
  type  & $\tau$ &class &   time (s)& time (s) & \#Opt&   time (s)& time (s) & \#Opt&   time (s)& time (s) & \#Opt\\										
									\hline		 
$a1$ & 0.5  & $\pi1$ & 6.65  & 100.00 & 33 & 0.01 & 0.05 & 35 & 0.01 & 0.03 & 35 \\
   &      & $\pi2$ & 16.33 & 100.00 & 30 & 0.02 & 0.14 & 35 & 0.01 & 0.12 & 35 \\
   &      & $\pi3$ & 19.56 & 100.00 & 29 & 0.08 & 0.89 & 35 & 0.02 & 0.14 & 35 \\
   &      & $\pi4$ & 21.98 & 100.00 & 28 & 0.03 & 0.39 & 35 & 0.01 & 0.06 & 35 \\
   &      & $\pi5$ & 22.95 & 100.00 & 28 & 0.05 & 0.38 & 35 & 0.01 & 0.10 & 35 \\
   &      & $\pi6$ & 25.89 & 100.00 & 27 & 0.04 & 0.89 & 35 & 0.01 & 0.05 & 35 \\
   &      & $\pi7$ & 15.49 & 100.00 & 30 & 0.09 & 1.34 & 35 & 0.01 & 0.03 & 35 \\
   &      & $\pi8$ & 11.73 & 100.00 & 31 & 0.06 & 0.63 & 35 & 0.03 & 0.40 & 35 \\
   & 0.1  & $\pi1$ & 7.53  & 100.00 & 33 & 0.00 & 0.05 & 35 & 0.01 & 0.18 & 35 \\
   &      & $\pi2$ & 15.01 & 100.00 & 30 & 0.01 & 0.08 & 35 & 0.01 & 0.08 & 35 \\
   &      & $\pi3$ & 29.18 & 100.00 & 25 & 0.01 & 0.08 & 35 & 0.01 & 0.03 & 35 \\
   &      & $\pi4$ & 17.96 & 100.00 & 29 & 0.01 & 0.06 & 35 & 0.01 & 0.03 & 35 \\
   &      & $\pi5$ & 18.44 & 100.00 & 29 & 0.01 & 0.09 & 35 & 0.01 & 0.03 & 35 \\
   &      & $\pi6$ & 12.68 & 100.00 & 31 & 0.01 & 0.11 & 35 & 0.00 & 0.04 & 35 \\
   &      & $\pi7$ & 17.93 & 100.00 & 30 & 0.09 & 1.38 & 35 & 0.01 & 0.05 & 35 \\
   &      & $\pi8$ & 5.98  & 100.00 & 33 & 0.03 & 0.58 & 35 & 0.01 & 0.05 & 35 \\
   & 0.01 & $\pi1$ & 0.70  & 3.83   & 35 & 0.00 & 0.00 & 35 & 0.00 & 0.01 & 35 \\
   &      & $\pi2$ & 7.19  & 100.00 & 33 & 0.00 & 0.03 & 35 & 0.01 & 0.03 & 35 \\
   &      & $\pi3$ & 7.52  & 100.00 & 34 & 0.00 & 0.05 & 35 & 0.00 & 0.03 & 35 \\
   &      & $\pi4$ & 20.09 & 100.00 & 30 & 0.01 & 0.05 & 35 & 0.01 & 0.03 & 35 \\
   &      & $\pi5$ & 16.57 & 100.00 & 32 & 0.01 & 0.05 & 35 & 0.00 & 0.03 & 35 \\
   &      & $\pi6$ & 3.37  & 30.16  & 35 & 0.00 & 0.01 & 35 & 0.00 & 0.01 & 35 \\
   &      & $\pi7$ & 15.94 & 100.00 & 30 & 0.02 & 0.20 & 35 & 0.00 & 0.01 & 35 \\
   &      & $\pi8$ & 0.70  & 8.34   & 35 & 0.01 & 0.09 & 35 & 0.00 & 0.01 & 35 \\
$a2$ & 0.5& $\pi1$ & 10.77 & 100.00 & 32 & 0.02 & 0.14 & 35 & 0.02 & 0.17 & 35 \\
   &      & $\pi2$ & 16.60 & 100.00 & 31 & 0.08 & 1.07 & 35 & 0.07 & 1.40 & 35 \\
   &      & $\pi3$ & 22.18 & 100.00 & 29 & 0.42 & 4.57 & 35 & 0.04 & 0.32 & 35 \\
   &      & $\pi4$& 19.38 & 100.00 & 30 & 0.51 & 7.81 & 35 & 0.03 & 0.21 & 35 \\
   &      & $\pi5$ & 21.57 & 100.00 & 29 & 0.28 & 3.00 & 35 & 0.03 & 0.28 & 35 \\
   &      & $\pi6$ & 11.32 & 100.00 & 32 & 0.47 & 8.02 & 35 & 0.03 & 0.23 & 35 \\
   &      & $\pi7$ & 9.68  & 100.00 & 32 & 0.49 & 7.75 & 35 & 0.02 & 0.11 & 35 \\
   &      & $\pi8$ & 14.59 & 100.00 & 30 & 0.05 & 0.80 & 35 & 0.01 & 0.09 & 35 \\
   & 0.1  & $\pi1$ & 14.78 & 100.00 & 32 & 0.02 & 0.17 & 35 & 0.05 & 0.38 & 35 \\
   &      & $\pi2$ & 21.28 & 100.00 & 28 & 0.04 & 0.30 & 35 & 0.02 & 0.22 & 35 \\
   &      & $\pi3$ & 18.82 & 100.00 & 29 & 0.22 & 2.16 & 35 & 0.03 & 0.27 & 35 \\
   &      & $\pi4$ & 30.67 & 100.00 & 25 & 0.28 & 2.17 & 35 & 0.03 & 0.13 & 35 \\
   &      & $\pi5$ & 25.16 & 100.00 & 27 & 0.22 & 1.98 & 35 & 0.02 & 0.12 & 35 \\
   &      & $\pi6$ & 28.61 & 100.00 & 26 & 0.19 & 1.63 & 35 & 0.02 & 0.09 & 35 \\
   &      & $\pi7$ & 10.08 & 100.00 & 32 & 0.26 & 7.30 & 35 & 0.02 & 0.11 & 35 \\
   &      & $\pi8$ & 11.68 & 100.00 & 31 & 0.04 & 0.44 & 35 & 0.01 & 0.07 & 35 \\
   & 0.01 & $\pi1$ & 1.60  & 12.74  & 35 & 0.00 & 0.01 & 35 & 0.01 & 0.03 & 35 \\
   &      & $\pi2$ & 28.11 & 100.00 & 28 & 0.01 & 0.14 & 35 & 0.01 & 0.04 & 35 \\
   &      & $\pi3$ & 26.29 & 100.00 & 28 & 0.15 & 0.73 & 35 & 0.02 & 0.05 & 35 \\
   &      & $\pi4$ & 19.68 & 100.00 & 31 & 0.11 & 1.17 & 35 & 0.02 & 0.05 & 35 \\
   &      & $\pi5$ & 33.97 & 100.00 & 25 & 0.10 & 0.67 & 35 & 0.02 & 0.05 & 35 \\
   &      & $\pi6$ & 24.26 & 100.00 & 29 & 0.12 & 0.84 & 35 & 0.02 & 0.06 & 35 \\
   &      & $\pi7$ & 17.34 & 100.00 & 30 & 0.02 & 0.10 & 35 & 0.01 & 0.03 & 35 \\
   &      & $\pi8$ & 11.77 & 100.00 & 31 & 0.02 & 0.16 & 35 & 0.00 & 0.02 & 35 \\ \hline
\end{tabular}}
		\caption{Computational results for instances with 1000 items and different correlations between penalties and weights: time (s) and number of optima over 35 instances.} 
		\label{tab:CPUTimePKP8}
\end{table}

\begin{table}[H]
	\centering
	\scriptsize
		\scalebox{0.85}{
\begin{tabular}{|ccc|*{3}{c|}*{3}{c|}*{3}{c|}}
  \hline 
    $n = 10000$  &  & & \multicolumn{3}{|c|}{CPLEX 12.5} & \multicolumn{3}{|c|}{Algorithm in \cite{CeRi06}} & \multicolumn{3}{|c|}{Proposed exact approach}  \\ \hline
  Weight & &Penalty &   Average &  Max  &&   Average &  Max  &&   Average &  Max  &\\
  type  & $\tau$ &class &   time (s)& time (s) & \#Opt&   time (s)& time (s) & \#Opt&   time (s)& time (s) & \#Opt\\										
									\hline		 
$a1$ & 0.5  & $\pi1$ & 26.96 & 100.00 & 29 & 1.07  & 8.74   & 35 & 1.04 & 10.54 & 35 \\
   &      & $\pi2$ & 51.68 & 100.00 & 18 & 4.93  & 100.00 & 34 & 0.89 & 8.74  & 35 \\
   &      & $\pi3$ & 45.50 & 100.00 & 20 & 8.83  & 100.00 & 33 & 1.06 & 5.78  & 35 \\
   &      & $\pi4$ & 52.72 & 100.00 & 17 & 24.28 & 100.00 & 28 & 2.22 & 17.67 & 35 \\
   &      & $\pi5$ & 54.65 & 100.00 & 17 & 18.01 & 100.00 & 29 & 1.66 & 10.09 & 35 \\
   &      & $\pi6$ & 55.55 & 100.00 & 16 & 23.99 & 100.00 & 27 & 2.46 & 20.65 & 35 \\
   &      & $\pi7$ & 38.21 & 100.00 & 25 & 11.72 & 100.00 & 31 & 1.08 & 10.73 & 35 \\
   &      & $\pi8$ & 34.49 & 100.00 & 27 & 7.54  & 100.00 & 33 & 9.05 & 70.71 & 35 \\
   & 0.1  & $\pi1$ & 27.19 & 100.00 & 28 & 0.88  & 18.70  & 35 & 0.72 & 18.88 & 35 \\
   &      & $\pi2$ & 48.67 & 100.00 & 22 & 2.17  & 20.73  & 35 & 0.70 & 6.27  & 35 \\
   &      & $\pi3$ & 54.65 & 100.00 & 18 & 7.52  & 100.00 & 33 & 1.30 & 18.19 & 35 \\
   &      & $\pi4$ & 59.08 & 100.00 & 17 & 15.20 & 100.00 & 32 & 0.91 & 7.16  & 35 \\
   &      & $\pi5$ & 57.49 & 100.00 & 19 & 8.46  & 100.00 & 34 & 1.00 & 8.53  & 35 \\
   &      & $\pi6$ & 51.93 & 100.00 & 21 & 7.55  & 100.00 & 34 & 0.39 & 2.26  & 35 \\
   &      & $\pi7$ & 48.77 & 100.00 & 22 & 11.15 & 100.00 & 32 & 0.85 & 8.59  & 35 \\
   &      & $\pi8$ & 23.95 & 100.00 & 30 & 4.07  & 73.67  & 35 & 1.27 & 8.73  & 35 \\
   & 0.01 & $\pi1$ & 21.36 & 100.00 & 30 & 0.17  & 3.00   & 35 & 0.31 & 6.57  & 35 \\
   &      & $\pi2$ & 37.31 & 100.00 & 26 & 6.20  & 73.05  & 35 & 0.58 & 5.66  & 35 \\
   &      & $\pi3$ & 45.55 & 100.00 & 21 & 2.26  & 33.15  & 35 & 0.28 & 3.69  & 35 \\
   &      & $\pi4$ & 45.75 & 100.00 & 21 & 2.36  & 59.95  & 35 & 0.33 & 2.50  & 35 \\
   &      & $\pi5$ & 49.01 & 100.00 & 21 & 1.25  & 10.15  & 35 & 0.46 & 5.62  & 35 \\
   &      & $\pi6$ & 44.44 & 100.00 & 23 & 1.28  & 26.15  & 35 & 0.58 & 14.20 & 35 \\
   &      & $\pi7$ & 38.94 & 100.00 & 22 & 13.75 & 100.00 & 31 & 0.47 & 3.36  & 35 \\
   &      & $\pi8$ & 21.96 & 100.00 & 31 & 3.30  & 100.00 & 34 & 0.35 & 6.39  & 35 \\
$a2$ & 0.5& $\pi1$ & 50.50 & 100.00 & 18 & 4.70  & 44.55  & 35 & 2.91 & 17.61 & 35 \\
   &      & $\pi2$ & 51.30 & 100.00 & 18 & 16.08 & 100.00 & 30 & 4.55 & 53.64 & 35 \\
   &      & $\pi3$ & 46.81 & 100.00 & 21 & 19.17 & 100.00 & 29 & 5.14 & 42.97 & 35 \\
   &      & $\pi4$ & 55.72 & 100.00 & 19 & 14.48 & 100.00 & 32 & 4.80 & 58.60 & 35 \\
   &      & $\pi5$ & 52.02 & 100.00 & 19 & 20.40 & 100.00 & 30 & 4.84 & 29.21 & 35 \\
   &      & $\pi6$ & 47.85 & 100.00 & 21 & 18.55 & 100.00 & 29 & 4.29 & 39.81 & 35 \\
   &      & $\pi7$ & 37.68 & 100.00 & 24 & 11.97 & 100.00 & 32 & 3.61 & 25.55 & 35 \\
   &      & $\pi8$ & 27.78 & 100.00 & 26 & 12.14 & 100.00 & 33 & 4.02 & 26.06 & 35 \\
   & 0.1  & $\pi1$ & 46.84 & 100.00 & 21 & 11.89 & 100.00 & 33 & 3.47 & 41.73 & 35 \\
   &      & $\pi2$ & 49.39 & 100.00 & 21 & 15.23 & 100.00 & 34 & 3.55 & 40.59 & 35 \\
   &      & $\pi3$ & 57.48 & 100.00 & 18 & 21.44 & 100.00 & 29 & 4.78 & 28.76 & 35 \\
   &      & $\pi4$ & 60.68 & 100.00 & 18 & 19.71 & 100.00 & 31 & 3.62 & 19.90 & 35 \\
   &      & $\pi5$ & 73.48 & 100.00 & 10 & 23.79 & 100.00 & 29 & 4.05 & 14.26 & 35 \\
   &      & $\pi6$ & 57.68 & 100.00 & 18 & 23.00 & 100.00 & 31 & 6.56 & 80.02 & 35 \\
   &      & $\pi7$ & 41.52 & 100.00 & 22 & 8.76  & 100.00 & 33 & 1.70 & 7.87  & 35 \\
   &      & $\pi8$ & 30.19 & 100.00 & 26 & 6.75  & 100.00 & 34 & 1.70 & 15.50 & 35 \\
   & 0.01 & $\pi1$ & 61.19 & 100.00 & 16 & 2.94  & 26.59  & 35 & 3.98 & 29.00 & 35 \\
   &      & $\pi2$ & 60.12 & 100.00 & 19 & 5.40  & 49.36  & 35 & 3.37 & 20.51 & 35 \\
   &      & $\pi3$ & 65.69 & 100.00 & 14 & 15.40 & 100.00 & 31 & 1.87 & 12.09 & 35 \\
   &      & $\pi4$ & 77.25 & 100.00 & 10 & 16.77 & 100.00 & 30 & 2.39 & 13.24 & 35 \\
   &      & $\pi5$ & 73.10 & 100.00 & 12 & 17.29 & 100.00 & 30 & 2.08 & 11.25 & 35 \\
   &      & $\pi6$ & 60.01 & 100.00 & 18 & 16.08 & 100.00 & 30 & 1.71 & 8.97  & 35 \\
   &      & $\pi7$ & 58.23 & 100.00 & 17 & 24.62 & 100.00 & 30 & 1.70 & 6.82  & 35 \\
   &      & $\pi8$ & 34.99 & 100.00 & 25 & 4.95  & 51.47  & 35 & 0.97 & 6.53  & 35 \\ \hline
\end{tabular}}
		\caption{Computational results for instances with 10000 items and different correlations between penalties and weights: time (s) and number of optima over 35 instances.} 
		\label{tab:CPUTimePKP9}
\end{table}

\fi{}

\section{Approximation results}
\label{PKP:APPROX}

%\subsection{Negative approximation result}
%\label{PKP:no_approx_algo}
In this section we investigate the approximability of PKP.
The classical 0--1 Knapsack Problem admits fully polynomial time approximation schemes (FPTAS), see, e.g.~\cite{KePfPi04}. 
PKP has ``only'' an additional penalty to consider in the objective with respect to KP. Thus, one might expect some straightforward approximation algorithm for this problem as well. Nonetheless, we prove here the general result that no polynomial time approximation algorithm exists for PKP (under $\mathcal{P}\neq\mathcal{NP}$).

\begin{theorem}\label{th:PKPnoapprox}
PKP does not admit a polynomial time algorithm with a bounded approximation ratio unless $\mathcal{P}=\mathcal{NP}$.
\end{theorem}

\begin{proof} 
The theorem is proved by reduction from the Subset Sum Problem (SSP). 
Given $n$ items with integer weights $w'_j$ ($j= 1,\dots,n$) and a value $W'$  (with $\sum_{j=1}^{n} w'_{j} > W'$), we recall that the decision version of SSP is an NP--complete problem and asks whether there exists a subset of items represented by $x^*$ such that $\sum_{j=1}^{n} w'_{j}x_{j}^*=W'$.

We build an instance of PKP with $n$ items, profits and weights $p_{j}=w_{j}=w_j'$, penalties $\pi_j = W' - 1$ ($j= 1,\dots,n$ ) and capacity $c=W'$. 
The capacity constraint implies that for every feasible solution there is 
$\sum_{j=1}^{n}p_jx_j=\sum_{j=1}^{n} w_jx_j \leq W'$. The penalty value will be equal to either $W' - 1$ if we pack at least one item or 0 otherwise, therefore the optimal solution of this PKP instance is bounded by $\sum_{j=1}^{n}p_jx_j - (W'-1) \leq 1$.
Not placing any item in the knapsack attains the trivial solution with value equal to $0$. By integrality of the input data, this limits the optimal solution value
to $0$ or $1$, where the latter value can be reached if and only if the Subset Sum Problem has a solution.

Hence, if there was a polynomial time algorithm for PKP with a bounded approximation ratio, we could decide SSP just by checking if the approximate solution of PKP is strictly positive. Clearly this is not possible unless $\mathcal{P} = \mathcal{NP}$. %under the assumption that $\mathcal{P} \neq \mathcal{NP}$.
\end{proof}

%\subsection{Positive approximation results}
%\label{sec:posapprox}

%\bigskip
%(To discuss: should we also insert somewhere the negative case $p_{max} \geq \pi_j$?)

While the result of Theorem~\ref{th:PKPnoapprox} rules out any reasonable approximation 
for the general case, one can impose mild restrictions on the input data which still permit
fully polynomial time approximation algorithms.
All our results are based on the following simple approximation algorithm:
Similar to the exact algorithm sketched in Section~\ref{sec:themodelPKP}
we consider all $n$ choices for the penalty value, namely $\Pi \in \{\pi_1, \ldots, \pi_n\}$.
For each choice $j$ of the leading item with $\Pi=\pi_j$,
we compute a suboptimal solution of problem $PKP_{j}^+$ by 
packing item $j$ into the knapsack and applying a $(1-\delta)-$approximation algorithm
for the remaining knapsack problem with capacity $c-w_j$ and item set $\{j+1,\ldots,n\}$.
The optimal solution value of the latter problem will be denoted as $z^R_j$.

As an output of the resulting approximation algorithm $A(\delta)$ with
objective function value $z^A(\delta)$ we
use the maximum value obtained over all $n$ iterations (including the empty set).
For a constant $\delta>0$, algorithm $A(\delta)$ can be performed by 
running $n$ times an FPTAS for KP.
%For the analysis we will denote the leading item of the optimal solution as $j^*$.
%In the analysis we will assume integer profits and penalties.
Note that if $z^*=0$, then also $A(\delta)$ will output a value of $0$.
Thus we can assume by integrality of the input data: 
\begin{equation}\label{eq:positive}
z^*=p_{j^*}+z^R_{j^*} - \pi_{j^*} \geq 1.
\end{equation}
As a general bound on the performance of $A(\delta)$ we get:
\begin{eqnarray*}
z^A(\delta) &\geq&
p_{j^*} + (1-\delta) z^R_{j^*} - \pi_{j^*} \\
&=&  
(1-\eps) z^* +(\eps-\delta)z^R_{j^*} - \eps (\pi_{j^*} - p_{j^*})
\end{eqnarray*}
Hence, we obtain an FPTAS for a suitable choice of $\delta \leq \eps$ if we can prove:
\begin{equation}\label{eq:claim}
(\eps-\delta)z^R_{j^*} \geq  \eps (\pi_{j^*} - p_{j^*})
\end{equation}
Note that whenever inequality $p_{j^*} \geq \pi_{j^*}$ is implied, condition (\ref{eq:claim}) is trivially satisfied for any $\delta \leq \eps$.\\
In the following we will describe four relevant cases 
%assumptions on the input data
which all permit an FPTAS for PKP. 
We start with the case where each item has a profit greater than (or equal to) its penalty value.
\begin{prop}\label{case2}
If $p_j \geq \pi_j$ for $j=1, \ldots, n$, then
algorithm $A(\delta)$ is an FPTAS for PKP.
\end{prop}

\begin{proof} 

%\bigskip
%{\bf Profits $\geq$ Penalties:}
%We assume that $p_j \geq \pi_j$ for all $j=1, \ldots, n$.
%Now 
It is trivial to see that setting $\delta:=\eps$ 
the right-hand side of (\ref{eq:claim}) is always less than (or equal to) zero
for every $j$ and thus the inequality is always fulfilled.
\end{proof} 

We henceforth assume the restricting condition $\pi_{j^*} - p_{j^*} \geq 0$. 
We first consider the case where penalties are bounded by a given constant $C$.

\begin{prop}\label{case3}
If $\pi_j + 1 \leq C$ for $j=1, \ldots, n$ and a constant $C$, algorithm $A(\delta)$ constitutes an FPTAS.
\end{prop}

\begin{proof} 
%\bigskip
%{\bf Bounded penalties:}
We can choose $\delta := \frac{\eps}{C}$ and consider from (\ref{eq:positive}) that $z^R_{j^*} \geq \pi_{j^*} - p_{j^*} + 1$. 
%We assume that there is a given constant $C$ with 
%In particular, $\pi_j +1 \leq C$.
%Choosing $\delta \leq \frac{\eps}{C}$ and plugging in (\ref{eq:positive}), 
%it can be shown that algorithm $A(\delta)$ is an FPTAS for PKP.
We have that:
\begin{eqnarray}
(\eps-\delta)z^R_{j^*} 
&\geq &
(\eps-\delta)(\pi_{j^*} - p_{j^*} + 1)\\
&=&\eps(\pi_{j^*} - p_{j^*} + 1) - 
\frac{\eps}{C} (\pi_{j^*} - p_{j^*} + 1) \label{eq:bounddelta}\\
&\geq&\eps(\pi_{j^*} - p_{j^*}) +\eps - 
\frac{\eps}{C} (C - p_{j^*} + 1) \\
&=& \eps(\pi_{j^*} - p_{j^*}) + 
\frac{\eps}{C} (p_{j^*} - 1) \\
&\geq&  \eps(\pi_{j^*} - p_{j^*}) \label{eq:verylast}
\end{eqnarray}
%For equality (\ref{eq:bounddelta}) we assume that $\pi_{j^*} - p_{j^*} + 1 \geq 0$ holds, because otherwise the above case applies. 
The last inequality (\ref{eq:verylast}) follows from the integrality of profits. Hence, condition (\ref{eq:claim}) is shown.
\end{proof} 

%\bigskip
%{\bf Penalties minus Profits is bounded:}
As generalization of the case in Proposition \ref{case2}, we consider for each item smaller profits than penalties and manage to derive an FPTAS as long as this difference is bounded by a constant.
\begin{prop}\label{case4}
PKP admits an FPTAS if $\pi_j - p_j \leq C$ for $j=1, \ldots, n$ and a constant $C$.
\end{prop}
\begin{proof} 
%we can still derive an FPTAS if profits are smaller than penalties, as long as this difference is bounded by a constant.
%Formally, we assume that there is a constant $C$ such that $\pi_j - p_j \leq C$ for all $j=1, \ldots, n$.
%As before, we restrict our attention to the case $\pi_{j^*} - p_{j^*} \geq 0$.
By choosing $\delta:= \frac{\eps}{C+1}$ we get:
\begin{eqnarray}
(\eps-\delta)z^R_{j^*} 
&=&
(\eps-\frac{\eps}{C+1})z^R_{j^*} \\
&\geq&
(\eps-\frac{\eps}{\pi_{j^*} - p_{j^*}+1})z^R_{j^*} \\
&\geq&
(\eps-\frac{\eps}{\pi_{j^*} - p_{j^*}+1})(\pi_{j^*} - p_{j^*} + 1) \label{eq:secIn} \\
&=& 
\eps(\pi_{j^*} - p_{j^*}) + \eps - \eps
\end{eqnarray}
For inequality (\ref{eq:secIn}) we invoke again (\ref{eq:positive}). This shows condition (\ref{eq:claim}).
\end{proof}

%\bigskip
%{\bf Profits and Penalties in a bounded interval:}

%Let us consider now 
Finally, consider the case where there exists a bounded interval containing all profit and penalty values.
This can be expressed by assuming a constant parameter $\rho \in (1/2,1)$ with
$p_{\min} \geq \rho\cdot \pi_{\max}$, i.e.\ all profits are not too small 
compared to the largest penalty. On the other hand, profits can well be arbitrarily large. The following proposition holds.
\begin{prop}\label{case5}
There is an FPTAS for PKP if $p_{\min} \geq \rho\cdot \pi_{\max}$ with $\rho \in (1/2,1)$.
\end{prop}
\begin{proof} 
If the optimal solution consists only of the leading item, 
then $A(\delta)$ also yields the optimum.
Thus, we have $z^R_{j^*} \geq p_{\min} \geq \rho\, \pi_{\max}$.
Choosing
$\delta := \eps\frac{2\rho-1}{\rho}$ (note that $\delta>0$ for $\rho > 1/2$) 
we get:
\begin{eqnarray}
(\eps-\delta)z^R_{j^*} 
&\geq&
(\eps-\delta) \rho\, \pi_{\max}\\
&=&
(\eps-\eps\frac{2\rho-1}{\rho}) \rho\, \pi_{\max}\\
&=&
(\eps\rho-2\eps\rho +\eps) \pi_{\max}\\
&=&
\eps(1-\rho) \pi_{\max}\\
&\geq&
\eps (\pi_{\max} - p_{\min}) \\
&\geq &
\eps(\pi_{j^*} - p_{j^*}) 
\end{eqnarray}
This shows again condition (\ref{eq:claim}).
\end{proof}

%\bigskip
%{\bf Comment on $A(\delta)$}:
%\begin{remark}
We remark that, although we cannot exclude other approximation schemes for PKP, it seems hard to construct any meaningful approximation algorithm different from $A(\delta)$.
%It is clear, that there exists a leading item $j^*$ in the optimal solution and thus considering 
The algorithm considers each term $p_{j} -\pi_{j}$ explicitly  
and thus will also include the part of the optimal solution value contributed by 
the leading item $p_{j^*} -\pi_{j^*}$. 
Then, the knapsack problem $z^R_{j^*}$ is solved by a $(1-\delta)-$approximation with a suitably chosen parameter $\delta$ 
which is the best one can do for the remaining sub--problem.
%that appears to be the best option for approximating such standard KP. 
%What can we do better than a $(1-\delta)-$approximation with a suitably chosen parameter $\delta$?
%However, it is hard to state that ``this is the best possible approximation''...
%\end{remark}

\section{Conclusions}
\label{sec:PKPconclusion}
We proposed a dynamic programming based exact approach for PKP which leverages an algorithmic framework originally constructed for KP. The proposed approach turns out to be very effective in solving instances of the problem with up to 10000 items and favorably compares to both solver CPLEX 12.5 and an exact algorithm in the literature. 
From a theoretical point of view we also show that PKP can be solved in the same
pseudopolynomial running time $O(nc)$ as the standard knapsack problem.
%thus improving upon an earlier $O(n^2 c)$ algorithm.
We also gave further insights on the structure and properties of PKP by providing a characterization of its linear relaxation and an effective procedure to compute upper bounds on the problem.
By studying the approximability of PKP, we showed rather surprisingly that
there is no polynomial time approximation algorithm with bounded approximation ratio, while imposing some mild conditions on the input of PKP allows an FPTAS.
In future research, we will investigate extensions of our procedures to other KP generalizations. It would also be interesting to evaluate the performances of our approach on new benchmark and challenging PKP instances. % \ref{appendixA:Proofs}.

\subsection*{Acknowledgments}

We thank the authors of \cite{CeRi06} for providing us with the code of their algorithm and the generation scheme of the instances.\\

Rosario Scatamacchia was supported by a fellowship from TIM Joint Open Lab SWARM (Turin, Italy).
Ulrich Pferschy was supported by the University of Graz project "Choice-Selection-Decision".

%\section*{References}
%\label{sec:bib}

\bibliographystyle{plain}  
\bibliography{Ref}

\end{document}